\documentclass[10pt,reqno,a4paper,nofootinbib]{article}
\usepackage{amsmath}
\usepackage{amsthm}
\usepackage{latexsym}
\usepackage{amsfonts}
\usepackage{amssymb}
\usepackage{color}
\usepackage{bbm}
\usepackage{dsfont}
\usepackage{graphicx}
\usepackage{textcomp}

\usepackage{authblk}

\newtheorem{theorem}{Theorem}
\newtheorem{proposition}{Proposition}

\newtheorem{corollary}{Corollary}

\theoremstyle{definition}
\newtheorem{definition}{Definition}

\newtheorem{remark}{Remark}
\newtheorem{example}{Example}

\newcommand{\be}{\begin{equation}}
\newcommand{\ee}{\end{equation}}
\newcommand{\bea}{\begin{eqnarray}}
\newcommand{\eea}{\end{eqnarray}}

\newcommand{\nat}{{\mathbb N}} 

\newcommand{\hi}{\mathcal{H}} 
\newcommand{\ket}[1]{|#1\rangle} 
\newcommand{\bra}[1]{\langle#1|} 
\newcommand{\no}[1]{||#1||} 
\newcommand{\tr}[1]{\textrm{tr}\left[#1\right]} 
\newcommand{\1}{\mathbbm{1}}
\newcommand{\0}{O} 

\newcommand{\cP}{\mathcal{P}}

\newcommand{\id}{{\rm{id}}}
\newcommand{\R}{\mathbbm{R}}
\newcommand{\cB}{{\mathcal{B}}}
\newcommand{\cA}{{\mathcal{A}}}

\newcommand{\cC}{{\cal C}}

\newcommand{\ii}{\mathbbm{1}}
\newcommand{\cM}{{\mathcal M}}
\newcommand{\ra}{\rightarrow}

\newcommand{\C}{\mathbb{C}}

\def\>{{\rangle}}
\def\<{{\langle}}

\newcommand{\cS}{{\mathcal{S}}}

\title{\Large\textbf{Extending quantum operations}}
\author[1]{Teiko Heinosaari\thanks{teiko.heinosaari@utu.fi}}
\author[2]{Maria A.~Jivulescu\thanks{maria.jivulescu@mat.upt.ro}}
\author[3]{David Reeb\thanks{david.reeb@tum.de}}
\author[3]{Michael M.~Wolf\thanks{m.wolf@tum.de}}
\affil[1]{\small{Turku Centre for Quantum Physics, Department of Physics and Astronomy, University of Turku, Finland}}
\affil[2]{\small{Department of Mathematics, University Politehnica Timi\c{s}oara, 300006 Timi\c{s}oara, Romania}}
\affil[3]{\small{Department of Mathematics, Technische Universit\"at M\"unchen, 85748 Garching, Germany}}

\date{October 22, 2012}

\begin{document}

\maketitle

\vspace{-0.6cm}\begin{abstract}
For a given set of input-output pairs of quantum states or observables, we ask the question whether there exists a physically implementable transformation that maps each of the inputs to the corresponding output. The physical maps on quantum states are trace-preserving completely positive maps, but we also consider variants of these requirements. We generalize the definition of complete positivity to linear maps defined on arbitrary subspaces, then formulate this notion as a semidefinite program, and relate it by duality to approximative extensions of this map. This gives a characterization of the maps which can be approximated arbitrarily well as the restriction of a map that is completely positive on the whole algebra, also yielding the familiar extension theorems on operator spaces. For quantum channel extensions and extensions by probabilistic operations we obtain semidefinite characterizations, and we also elucidate the special case of Abelian in- or outputs. Finally, revisiting a theorem by Alberti and Uhlmann, we provide simpler and more widely applicable conditions for certain extension problems on qubits, and by using a semidefinite programming formulation we exhibit counterexamples to seemingly reasonable but false generalizations of the Alberti-Uhlmann theorem.
\end{abstract}

\section{Introduction}\label{sec:intro}

In this work we study the possibility of extending partially given quantum operations.

The paradigmatic case is the task where we need to design a device that transforms quantum states into quantum states and satisfies the input-output conditions
\be\rho_i\,\mapsto\,\rho'_i\,,\quad i=1,\ldots,N~,\label{eq:criteria-states}\ee
for some finite number of given pairs $(\rho_i,\rho'_i)$ of quantum states ($\rho_i\in\cM_d$, $\rho'_i\in\cM_{d'}$). Is it possible to construct this kind of device?

If there exists a quantum channel $T:\cM_d\ra\cM_{d'}$, i.e.~a trace-preserving completely positive linear map on the full algebra $\cM_d$, such that
\be\label{eq:criteria-states-2}T(\rho_i)\,=\,\rho'_i\,,\quad i=1,\ldots,N~,\ee
then the task is feasible \cite{nielsenchuang}. Indeed, every quantum channel can be implemented physically as a unitary evolution on a composite system.
It is easy to invent examples of both types, possible and impossible such tasks, and the question is therefore meaningful. Note that $T$ in (\ref{eq:criteria-states-2}) is an extension of the map $\rho_i\mapsto\rho'_i$, defined on the subspace ${\rm span}\{\rho_i\}$, to the whole space $\cM_d$. Therefore the above problem can be phrased as an extension question for quantum channels or for completely positive maps more generally. In this work we restrict to the finite-dimensional case.

Another way to look at the question of the existence of a quantum channel achieving (\ref{eq:criteria-states-2}) is the following: If, in a tomographic experiment, the time-evolution (\ref{eq:criteria-states}) has been measured, then could this evolution have possibly been Markovian? If the evolution was Markovian, i.e.~it did neither depend on the preparation procedure of the input state nor on the memory of the environment, then there \emph{necessarily} exists a quantum channel $T$ satisfying (\ref{eq:criteria-states-2}). Therefore, a negative answer to the existence question falsifies the hypothesis of a Markovian evolution, and our treatment of the question (e.g.~in subsection \ref{cptpsubsection}) will offer a way to quantify this failure, complementing the analysis in \cite{assessing,breuermarkov}.

As mentioned before, the feasibility of the task \eqref{eq:criteria-states} is equivalent to the existence of a quantum channel satisfying \eqref{eq:criteria-states-2}. However, we can further generalize the problem. Suppose that there exists a completely positive map $T$ satisfying \eqref{eq:criteria-states-2}, but $T$ is not necessarily trace-preserving. Then we can define two completely positive trace-non-increasing maps by
\be T_1(\rho)~:=~T(\rho)/||T^*(\ii)|| \, , \quad T_0(\rho)~:=~\tr{(\1-T_1^*(\1))\rho} \eta_0~,\nonumber\ee
which add up to a trace-preserving map; here, $\eta_0$ is any fixed quantum state, $T^\ast$ is the dual map of $T$, and $||\cdot||$ here and in the following denotes the operator norm. These two maps thus form an instrument \cite{QTOS76} and describe a measurement with two outcomes.
If we perform a measurement on one of our input states $\rho_i$, then the transformation $\rho_i\mapsto\rho'_i$ takes place whenever we obtain the measurement outcome $1$, which happens with probability $\tr{T_1(\rho_i)}=1/||T^*(\ii)||$, independent of $i$. We conclude that any completely positive map $T$ satisfying \eqref{eq:criteria-states-2} gives rise to a probabilistic device solving the task. This motivates to study not only quantum channels but all completely positive maps. (In subsection \ref{subsectprobabextensions} we will look at more general probabilistic operations with possibly distinct success probabilities.)

Another aspect comes from the fact that completely positive maps, possibly supplemented by the requirement of (sub-)unitality, describe time-evolution in the Heisenberg picture (we elaborate on this interpretation in subsection \ref{sec:interpretation}). In this context the evolution acts on observables, which are selfadjoint operators.

We will therefore consider the outlined problem mostly in a general form, where the inputs and ouputs are two finite collections of selfadjoint operators, not necessarily quantum states as in (\ref{eq:criteria-states-2}). If the inputs $\{X_i\}_{i=1}^N$ and outputs $\{Y_i\}_{i=1}^N$ are specified, as one of our main contributions we ask and answer in Sections \ref{sectapproxextensions} and \ref{CPextensionssect} the following questions:
\begin{itemize}
\item[{\it(Q1)}]Does there exist a completely positive linear map ${T}:\cM_d\to\cM_{d'}$ such that ${T}(X_i)=Y_i$ for all $i=1,\ldots,N$\,?
\item[{\it(Q2)}]Can one achieve the mapping in {\it(Q1)} approximatively, i.e.~does there exist a sequence of completely positive maps ${T}_k:\cM_d\to\cM_{d'}$ such that
\be\lim_{k\to\infty}||{T}_k(X_i)-Y_i||\,=\,0\quad \text{for all}~\,i=1,\ldots,N~\textrm{?}\nonumber\ee
\end{itemize}
Clearly, {\it(Q1)} can have an affirmative answer only if {\it(Q2)} has. Perhaps surprisingly, the answer to {\it(Q2)} can be affirmative even if {\it(Q1)} has a negative answer. One of our main contributions is do delineate the difference between {\it(Q1)} and {\it(Q2)}, and to characterize case {\it(Q2)} exactly (Theorem \ref{th:cpisapprox}).

If the inputs $\{X_i\}$ span an operator system, then question {\it(Q1)} is well-studied in the theory of operator algebras, its answer being Arveson's extension theorem \cite{arvesonextension,paulsen}. Beyond operator systems, in this paper we examine these questions for arbitrary (Hermitian) operator spaces, as motivated above. 
One of our main tools will be semidefinite programming techniques \cite{rockafellar,convexoptimization} and thus the Hahn-Banach Theorem. In order to answer question {\it(Q2)} (in Theorem \ref{th:cpisapprox}), however, we have to apply strong duality in the direction \emph{opposite} to its ordinary application for Arveson's extension theorem (see Remark \ref{dualityinoppositedirection}).

More specifically, extension questions for physical maps on state spaces have been studied both for the classical case of probability distributions (e.g.~\cite{ruchetal,torgersenbook}) and for the quantum case (e.g.~\cite{albertiuhlmann}), and both are related to the theory of majorization (e.g.~\cite{connectmajorization}). In the quantum case, a convenient criterion in terms of Gram matrices is known for pure input and output states, and partially also in the cases where one of both sets contains mixed states \cite{cjw}. Recent work in the same direction as ours is \cite{jencovacp}, which on the one hand limits consideration to completely positive extensions of maps defined on subsets of the quantum states (obtaining our Corollary \ref{corollaryspannedbydensitymatrices}), but on the other characterizes the set of all such extensions. The question of which measurements (POVMs) are not irreversibly connected to another measurement by a quantum channel (in the Heisenberg picture) was investigated in \cite{cleanpovms}. Another recent work on the quantum side \cite{commutingquantum} concerned commuting in- and outputs. Extension problems of a similar type appear in (classical and quantum) statistical decision theory, in relation with the notion of sufficiency \cite{torgersenbook,buscemicmp}; the results obtained in this area (e.g.~\cite{buscemicmp,matsumotorandomizationpaper,jencovamorph}) are similar in spirit to, and serve as a motivation to, our results in Section \ref{albertiuhlmannsection}.

\medskip

This article is organized as follows. In Section \ref{sec:basic} we will introduce basic notions, in particular complete positivity of maps on arbitrary subspaces and its operational interpretation. In Section \ref{CPnesssect} we formulate, for a given map, the condition of complete positivity quantitatively as a semidefinite program (Theorem \ref{th:CPnessviaSDPtheorem}). In Section \ref{sectapproxextensions} we relate this condition via strong duality to the existence of completely positive approximative extensions of the map, which allows us to answer question {\it(Q2)} above, one of our main results (Theorem \ref{th:cpisapprox}). We further demonstrate that, maybe surprisingly, {\it(Q1)} and {\it(Q2)} are in general different questions. In Section \ref{CPextensionssect} we find answers to question {\it(Q1)} above: We recover slight generalizations of Arveson's extension theorem (Theorem \ref {th:exactcpextension}, cf.~\cite{arvesonextension,paulsen}; Corollary \ref{corollaryspannedbydensitymatrices}, cf.~\cite{jencovacp}) and connect it to question {\it(Q2)} above, and we find an independent criterion in subsection \ref{subspacewopositiveoperators}. Variations of the quest for completely positive maps are discussed in Section \ref{variationssect}: quantum channel extensions in subsection \ref{cptpsubsection}, which we can only characterize by a semidefinite program, and probabilistic operations and their connection to unambiguous state discrimination in subsection \ref{subsectprobabextensions}. Extension questions are known to simplify for commuting domain or range \cite{paulsen}, which we elaborate on in subsection \ref{commutativesection}.

In Section \ref{albertiuhlmannsection} we treat a special case of the extension question, first investigated by Alberti and Uhlmann \cite{albertiuhlmann} for the quantum case. Our first main result here (Theorem \ref{ourfidelitycrition}) is a necessary and sufficient condition for the existence of a quantum channel extension. This criterion is much easier to check and has wider applicability than the original Alberti-Uhlmann criterion. Secondly, in subsection \ref{counterexpsection} we provide counterexamples to a supposed generalization (that holds for classical probability distributions \cite{ruchetal}) of the Alberti-Uhlmann Theorem. To the best of our knowledge, these are the first such counterexamples found in the literature, and they also illustrate the semidefinite programming techniques used in the earlier sections.

\section{Basic concepts and preliminary observations}\label{sec:basic}

In this work we are concerned with linear maps $T:\cS\ra\cM_{d'}$ defined on a subspace $\cS\subseteq\cM_d$ of the $d$-dimensional matrix algebra $\cM_d\equiv\cM_d(\C)$. Our task is to extend such maps to all of $\cM_d$ while satisfying certain requirements. We restrict to $1\leq d, d'<\infty$.

\subsection{Complete positivity}\label{sec:cp}

Whereas the notion of complete positivity is often defined only for maps on the whole matrix algebra or for maps on operator systems \cite{paulsen}, it has a direct generalization to maps on subspaces. (By a subspace we will always mean a linear subspace.)

\begin{definition}[$n$-positive map]\label{definepositivemap}A linear map $T:\cS\ra\cM_{d'}$ on a subspace $\cS\subseteq\cM_d$ is \emph{$n$-positive} (for some $n\in\nat$) if for all positive operators $P\in\cS\otimes\cM_n$ we have $(T\otimes\id_n)(P)\geq0$. A $1$-positive map is called \emph{positive}.\end{definition}
Throughout, we consider subspaces such as $\cS\subseteq\cM_d$ or $\cS\otimes\cM_n\subseteq\cM_d\otimes\cM_n$ embedded in a matrix algebra, and we call an element of a subspace \emph{positive} if it is positive when regarded as an element of the full matrix algebra. (This is thus not to be confused with the elements that can be written as $A^\dagger A$ with $A$ from the subspace.)

\begin{definition}[Completely positive map]\label{def:CPmaponS}
A linear map $T:\cS\ra\cM_{d'}$ on a subspace $\cS\subseteq\cM_d$ is \emph{completely positive} if $T$ is $n$-positive for all $n\in\nat$.
\end{definition}

Given a completely positive map $T:\cS\ra\cM_{d'}$, we will often be asking for a completely positive extension $\widetilde{T}:\cM_d\ra\cM_{d'}$, meaning that $\widetilde{T}(X)=T(X)$ for all $X\in\cS$. Since the domain of $\widetilde{T}$ contains the identity operator $\ii$, it follows from the positivity that $\widetilde{T}$ is \emph{Hermiticity-preserving}, i.e.~$\widetilde{T}(X^\dagger)=\widetilde{T}(X)^\dagger$ for all $X\in\cM_d$. Thus, we can without loss of generality consider $\cS$ to be a \emph{Hermitian subspace} (i.e.~$\cS=\cS^\dagger$) and require $T$ to be Hermiticity-preserving in the first place. Note, however, that $\cS$ need not contain the identity matrix $\ii$ or any other full-rank matrix.

A practical way to test whether a linear map $T:\cM_d\ra\cM_{d'}$ is completely positive is via its \emph{Choi representation}. The Choi isomorphism \cite{Choi75} sets a one-to-one correspondence between linear maps $T:\cM_d\ra\cM_{d'}$ (defined on the full matrix algebra) and operators $C\in\cM_{d'}\otimes\cM_d$, and is defined by
\begin{equation}\label{choicorrespondenceeqn}
 C=(T\otimes\id_d)(\ket{\Omega}\bra{\Omega})\qquad\textrm{and}\qquad T(X)=\mathrm{tr}_2[C(\ii\otimes X^T)] \, ,
 \end{equation}
where $\ket{\Omega}=\sum_{i=1}^d\ket{ii}$ is the \emph{(unnormalized) maximally entangled state} associated with a fixed orthonormal basis $\{\ket{i}\}_i$ of the Hilbert space $\C^d$, and $^T$ denotes the usual transposition w.r.t.~this orthonormal basis. A linear map $T$ defined on the whole matrix algebra $\cM_d$ is Hermiticity-preserving iff the corresponding \emph{Choi matrix} $C$ is Hermitian, and $T$ is completely positive iff $C$ is positive \cite{Choi75}.
The \emph{dual map} $T^*:\cM_{d'}\ra\cM_d$, defined by $\tr{Y^\dagger T(X)}=\tr{T^*(Y)^\dagger X}$ for all $X\in\cM_d,Y\in\cM_{d'}$, can be expressed in terms of the Choi matrix $C$ as
\begin{equation*}
T^*(Y)={\rm tr}_1[C^\dagger(Y\otimes\ii)]^T \, .
\end{equation*}

We emphasize that the Choi representation can be defined only if the linear map $T$ is given on the full matrix algebra $\cM_d$. One main result of Section \ref{CPnesssect} will establish an efficient procedure for checking complete positivity of any (Hermiticity-preserving) linear map $T$ on a subspace (Theorem \ref{th:CPnessviaSDPtheorem}, cf.~also Proposition \ref{CPdualityprop}). In the following we demonstrate that complete positivity can sometimes be verified directly from Definition \ref{def:CPmaponS}.

\begin{example}[Direct verification of complete positivity]\label{ex:easytocheckcp}
In some cases it is easy to check that a linear map is completely positive on $\cS$ even when $\cS\neq\cM_d$, so that the criterion of the positivity of the Choi matrix is not applicable. For instance, let $\cS:={\rm span}\{\Psi_0:=\ket{0}\bra{0},\sigma_x\}\subset\cM_2$, where $\sigma_x$ is the Pauli $x$-matrix, and define a linear map on $\cS$ via $T(\Psi_0):=P$ and $T(\sigma_x):=B$ with any fixed positive $P\geq0$ and any $B\in\cM_{d'}$.
Then $T$ is completely positive.
In order to see this, we first characterize all positive elements in $\cS\otimes\cM_n$.
An arbitrary $Q\in\cS\otimes\cM_n$ can be written in the form
\begin{equation*}
 Q\,=\,\Psi_0\otimes F_0+\sigma_x\otimes F_1
 \end{equation*}
for some $F_i\in\cM_n$. If $Q\geq 0$, then $F_0=F_0^\dagger$ and $F_1=F_1^\dagger$ due to $Q=Q^\dagger$, and for any vector $\ket{\psi}\in\C^n$ we must have
\be0~\leq~\left(\ii\otimes\bra{\psi}\right)Q\left(\ii\otimes\ket{\psi}\right)~=~\bra{\psi}F_0\ket{\psi}\Psi_0+\bra{\psi}F_1\ket{\psi}\sigma_x~.\nonumber\ee
The determinant of the last expression is $-\bra{\psi}F_1\ket{\psi}^2$, so that $Q$ can only be positive if $F_1=0$.
It follows that we must have $F_0\geq0$, so indeed $(T\otimes\id_n)(Q)=P\otimes F_0\geq0$.
\end{example}

Positive matrices with unit trace will be called \emph{(quantum) states}.
All quantum-mechanical time-evolutions are completely positive maps $T:\cM_d\ra\cM_{d'}$ by the requirement that they preserve the positivity of states even when applied to a subsystem $\cM_d$ of the whole system $\cM_d\otimes\cM_n$. More precisely, a \emph{quantum channel} is a completely positive map $T:\cM_d\ra\cM_{d'}$, defined on the whole matrix algebra, which is \emph{trace-preserving}, i.e.~$\tr{T(X)}=\tr{X}$ for all $X\in\cM_d$. We call a completely positive map $T:\cM_d\ra\cM_{d'}$ that is \emph{trace-non-increasing}, i.e.~$\tr{T(X)}\leq\tr{X}$ for all positive $X\in\cM_d$, a \emph{quantum operation}. These notions refer to the time-evolution of quantum states, commonly called the \emph{Schr\"odinger picture}. By the pairing $\tr{A^\dagger T(B)}=\tr{T^*(A)^\dagger B}$, the dual map $T^*$ then acts on quantum observables, called the \emph{Heisenberg picture}. $T^*$ is (completely) positive iff $T$ is, and $T$ is trace-preserving (resp.~trace-non-increasing) iff $T^*$ is \emph{unital}, i.e.~$T^*(\ii)=\ii$ (resp.~\emph{sub-unital}, i.e.~$T^*(\ii)\leq\ii$).

\medskip

We end this subsection with a small observation on subspaces without nonzero positive elements.
To motivate this class of situations, we consider an example. Suppose we want to know whether there exists a quantum operation which evolves the Pauli spin operators $\sigma_i\in\cM_2$ ($i=x,y,z$) in the Heisenberg picture into some other given observables $Y_i$  but which may not necessarily be a deterministic time-evolution, i.e.~may not be unital.
Obviously, such an evolution can exist only if the map $T:\sigma_i\mapsto Y_i$ defined on the subspace $\cS:={\rm span}\{\sigma_i\}_i$ is already completely positive. In this situation the question of complete positivity has a simple answer as the only positive operator in $\cS$ is $0$. Namely, we have the following general observation.
\begin{proposition}[Complete positivity on subspaces without positive elements]\label{prop:oposelementthencp}
If a subspace $\cS\subset\cM_d$ does not contain any nonzero positive operator, then every linear map $T:\cS\ra\cM_{d'}$ is completely positive.
\end{proposition}

\begin{proof}
Fix $n\in\nat$ and consider a positive element $P\in\cS\otimes\cM_n$.
The positivity of $P$ implies that also the partial trace over the second tensor factor is positive, $0\leq{\rm tr}_2[P]\in\cS$. Due to our assumption that $0$ is the only positive element in $\cS$ we conclude that ${\rm tr}_2[P]=0$, which also means that $\tr{P}=0$.
Thus $P$ is a positive matrix with vanishing trace, so $P=0$.
Finally, this implies $(T\otimes{\rm id}_n)(P)=0$, and therefore $T$ is $n$-positive.
As $n\in\nat$ was arbitary, $T$ is completely positive.
\end{proof}

\subsection{Linearity and Hermiticity preconditions}\label{sec:linearity}

Often we will have the situation that, given the matrices $X_i\in\cM_d$, $Y_i\in\cM_{d'}$, we are looking for a linear map $T$ that satisfies $T(X_i)=Y_i$, possibly along with some other criteria.
The existence of $T$ and of a linear extension $\widetilde{T}:\cM_d\ra\cM_{d'}$ to the whole space, can easily be checked by checking linear dependencies.
Thus, we will always assume that the input set $\{X_i\}_i$ is a linearly independent set and that $T$ is given as a linear map on $\cS:={\rm span}\{X_i\}_i$, which is then to be extended in a certain fashion.

Conversely, given a linear map $T:\cS\ra\cM_{d'}$ we may equivalently assume that $T$ is specified by its action $T(X_i)=Y_i$ on a basis $\{X_i\}_i$ of $\cS$.

In the Hermiticity-preserving case (see subsection \ref{sec:cp}), we may assume that all input operators $X_i$ (and thus also all output operators $Y_i$) are Hermitian matrices.
The action of a Hermiticity-preserving map is determined by its action on the real subspace of Hermitian matrices.

\subsection{Trace-preservation and compatibility}\label{sec:compatibility}
\label{linearityandcompatibilityprecond}
In many cases we are looking for extensions which are also trace-preserving. In the paradigmatic case we ask whether a linear map $T$ is a restriction of a quantum channel. Obviously, any map $T$ that is linear and trace-preserving on a subspace $\cS$ can be extended to a linear and trace-preserving map $\widetilde{T}$ on the whole space, namely by extending a basis of $\cS$ to a basis of $\cM_d$ and having $\widetilde{T}$ act in a trace-preserving way on the additional basis elements, e.g.~as the identity.

More generally, given two linearly independent sets $\{X_i\}\subset\cM_d$ and $\{X'_j\}\subset\cM_{d'}$, we can ask for the existence of a linear map $T:\cM_d\ra\cM_{d'}$ satisfying $T(X_i)=Y_i$ and $T^*(X'_j)=Y'_j$ for given $Y_i\in\cM_{d'}$ and $Y'_j\in\cM_d$.
We can thus have simultaneous constraints on the Schr\"odinger and Heisenberg pictures, respectively.

The question of a linear trace-preserving extension corresponds to the special case where $X'_1=\ii_{d'}$, $Y'_1=\ii_d$ is the only constraint on $T^*$.
It is easy to see from the definition of $T^*$ that such an extension exists iff the compatibility conditions $\tr{X'^\dagger_j Y_i}=\tr{Y'^\dagger_jX_i}$ hold for all $i,j$.

\subsection{Operational interpretation}\label{sec:interpretation}

We will now give an operational interpretation of a completely positive map $T:\cM_d\ra\cM_{d'}$ satisfying $T(X_i)=Y_i$ for $i=1,\ldots,N$. In the Introduction (Section \ref{sec:intro}) we already explained the interpretation if the operators are states (Schr\"odinger picture). Here we imagine $X_i,Y_i$ to be observables and $T$ to describe a time-evolution in the Heisenberg picture.

Analogous to the Introduction, we can rescale $T$ to a sub-unital completely positive map $T_1:=T/||T(\ii)||$, then find another completely positive map $T_0$ so that $T_1+T_0$ is unital, which implies that $(T_1^*,T_0^*)$ forms an instrument with two possible outcomes \cite{QTOS76}. Then, for all observables $Y_i$ from the given set and for all quantum states $\rho\in\cM_{d'}$ we have the identity
\bea\tr{\rho Y_i}~=~||T(\ii)||\,\tr{T_1^*(\rho)X_i}&=&||T(\ii)||\,\Big(\sum_{k=1}^d\lambda_k^{(i)}\tr{T_1^*(\rho)P_k^{(i)}}\,+\,0\cdot\tr{T_0^*(\rho)}\Big)\nonumber\\
&=&||T(\ii)||\,\Big(\sum_{k=1}^d\lambda_k^{(i)}\tr{\rho\,T_1(P_k^{(i)})}\,+\,0\cdot\tr{\rho\,T_0(\ii)}\Big)~,\nonumber\eea
where $X_i=\sum_k\lambda_k^{(i)}P_k^{(i)}$ is the spectral decomposition. The latter two forms can be interpreted as the following repeated measurement procedure: Put the quantum state $\rho$ into the instrument $(T_1^*,T_0^*)$ and, in case of the successful outcome $1$, measure the observable $X_i$. Record the value $0$ in case of the unsuccessful outcome $0$, and then average over all measurements. Irrepective of the input state $\rho$, the final result will equal the observable average $\tr{\rho Y_i}$ up to a scaling factor $||T(\ii)||$, which is independent of the observable $Y_i$ to be measured.

A large rescaling factor $||T(\ii)||\gg1$ in the previous identities does not necessarily mean that the map $T$ was chosen badly for measuring the desired observable average $\tr{\rho Y_i}$ -- in fact, if $T$ was a multiple of a unital map, then this amounts just to a uniform rescaling of the $X_i$ or, equivalently, of the $\lambda_k^{(i)}$. If, however, the ratio between the biggest and smallest eigenvalue of $T(\ii)$ is large, then there are some quantum states $\rho$ for which the instrument $(T_1^*,T_0^*)$ will often return the failing outcome $0$, so that many experimental runs may be necessary in order to measure the overservable average to some desired (relative) accuracy.

\section{Complete positivity on a subspace}\label{CPnesssect}
Whereas complete positivity of a map $T:\cM_d\ra\cM_{d'}$ between finite-dimensional matrix algebras can be checked easily through its associated Choi matrix $C=(T\otimes{\rm id})(\ket{\Omega}\bra{\Omega})$ \cite{Choi75}, no such criterion seems to be known for maps defined on a subspace $\cS\subset\cM_d$.
The following theorem provides the basis for a criterion that is algorithmically checkable with efficiency similar to Choi's criterion through semidefinite programming techniques \cite{convexoptimization}. For completely positive extension problems, the case of a Hermitian subspace $\cS=\cS^\dagger$ is particularly relevant.

\begin{theorem}[Complete positivity on a subspace]\label{th:CPnessviaSDPtheorem}
Let $T:\cS\ra\cM_{d'}$ be a linear map on a subspace $\cS\subseteq\cM_d$.
Let $\cS={\rm span}\{X_i\}_i$.
Then, $T$ is completely positive iff for any matrices $H_i\in\cM_{d'}$ the following implication holds:
\be
\sum_i X_i\otimes H_i\,\geq\,0\qquad\Rightarrow\qquad\sum_i{\rm{tr}}\left[T(X_i)^TH_i\right]\,\geq\,0~.\label{cpconditionintheorem}
\ee
If, in addition, $\cS$ is Hermitian and $T$ is Hermiticity-preserving, then by choosing the $X_i$'s Hermitian, complete positivity of $T$ can be checked via a semidefinite program with Hermitian variables $\{H_i\}_i$.
\end{theorem}
\begin{proof}
If $T$ is completely positive and $\sum_i X_i\otimes H_i\geq0$, then by Definition \ref{def:CPmaponS} we have
\be0\,\leq\,\bra{\Omega}\,\sum_i T(X_i)\otimes H_i\,\ket{\Omega}\,=\,\sum_i\tr{T(X_i)^TH_i}~,\nonumber\ee
where again $\ket{\Omega}:=\sum_{i=1}^{d'}\ket{ii}$ and we used that $\bra{\Omega}X\otimes Y\ket{\Omega}=\tr{X^TY}$.

Conversely, assume that for some $n\in\nat$ the matrix $P=\sum_i X_i\otimes K_i\in\cS\otimes\cM_n$ is positive.
Now, any $\ket{\psi}\in\hi_{d'}\otimes\hi_n$ can be written as $\ket{\psi}=(\ii\otimes R)\ket{\Omega}$ for some matrix $R\in\C^{n\times d'}$. Thus,
\be\bra{\psi}\,(T\otimes\id_n)(P)\,\ket{\psi}~=~\bra{\Omega}\,\sum_i T(X_i)\otimes(R^\dagger K_iR)\,\ket{\Omega}~=~\sum_i\tr{T(X_i)^T\,R^\dagger K_iR}~\geq~0\nonumber\ee
by condition \eqref{cpconditionintheorem}, since
\begin{equation*}
\sum_i X_i\otimes R^\dagger K_iR=(\ii\otimes R)^\dagger P(\ii\otimes R)\geq 0
\end{equation*}
as $P$ was positive.
As this holds for any $\bra{\psi}$, we have $(T\otimes{\rm id}_n)(P)\geq0$ and $T$ is completely positive.

If $\cS$ is Hermitian, one can choose a Hermitian basis $\{X_i\}_i$ of $\cS$ (cf.~beginning of the proof of Theorem \ref{th:cpisapprox}).
And, even more generally, if $\cS={\rm span}\{X_i\}_i$ with Hermitian $X_i$ then for any positive (and therefore, Hermitian) matrix $P\in\cS\otimes\cM_n$ there exist Hermitian matrices $H_i$ such that $P=\sum_i X_i\otimes H_i$. Further, if $T$ is Hermiticity-preserving, the matrices $T(X_i)^T$ are Hermitian. Thus, in this case condition (\ref{cpconditionintheorem}) involves only semidefinite inequalities and Hermitian variables, so that it can be checked by a semidefinite program.
This is made more explicit in the following discussion.
\end{proof}

Continuing the discussion from the last part of Theorem \ref{th:CPnessviaSDPtheorem}, if linearly independent Hermitian matrices $X_i\in\cM_d$, $Y_i\in\cM_{d'}$ ($i=1,\ldots,N$) are given, one can check complete positivity of the linear map $T:X_i\mapsto Y_i$ by the following semidefinite program (\ref{CPSDPminimizingfunction})--(\ref{CPSDPconditions}) (see \cite{convexoptimization} for a general introduction to semidefinite optimization problems):
\be\Gamma\,:=\,\inf~\sum_{i=1}^N \tr{Y^T_iH_i}\label{CPSDPminimizingfunction}~,\ee
where the infimum runs over all Hermitian matrices $H_i\in\cM_{d'}$ subject to the constraint
\be
\sum_{i=1}^N X_i\otimes H_i\,\geq\,0~.\label{CPSDPconditions}
\ee
Note in particular that this semidefinite program (SDP) has a feasible point $H_i\equiv0$. Due to homogeneity, this SDP can only assume the values $\Gamma=-\infty$ or $\Gamma=0$, and it is exactly the latter case in which the map $T:X_i\mapsto Y_i$ is completely positive.
(If the linear independence assumption on the $X_i$ is dropped, then one can still compute the SDP (\ref{CPSDPminimizingfunction})--(\ref{CPSDPconditions}) and it is easy to see that it will return the value $\Gamma=-\infty$ if $X_i\mapsto Y_i$ does not define a linear map.)

The feature of homogeneity allows to supplement the above SDP (\ref{CPSDPminimizingfunction})--(\ref{CPSDPconditions}) with an additional constraint, such as
\be||H_i||\,\leq\,1\quad\forall i~,\label{onepossiblecompactnesscondition}\ee
so that the optimization is bounded, the infimum $\Gamma$ in (\ref{CPSDPminimizingfunction}) will be attained, and the dichotomy between $\Gamma=0$ and $\Gamma=-\infty$ becomes a dichotomy between $\Gamma=0$ and $\Gamma<0$.

A condition similar to our Theorem \ref{th:CPnessviaSDPtheorem} for complete positivity of a map on an \emph{operator system} $\cS$ has been given as Theorem 6.1 in \cite{paulsen}, where however the Hahn-Banach theorem is used in order to first extend the map, which is not possible on general subspaces $\cS$, the case we consider (cf.~Theorem \ref{th:cpisapprox} below and subsequent examples). Furthermore, Theorem \ref{th:CPnessviaSDPtheorem} and the discussion above make the SDP formulation explicit, which is useful for concrete computations and to obtain the analytical insights on approximation questions described next.

\section{Approximative extensions}\label{sectapproxextensions}
Above, we were given Hermitian matrices $X_i$ and $Y_i$ and have asked whether or not the map $X_i\mapsto Y_i$ is linear and completely positive on $\cS={\rm span}\{X_i\}_i$. But in the same situation one can also, and seemingly quite differently, ask for the \emph{best approximation} of the linear map $T:X_i\mapsto Y_i$ by a completely positive linear map defined on the full algebra $\cM_d$. Recall from the introduction that if $\widetilde{T}:\cM_d\ra\cM_{d'}$ satisfies $\widetilde{T}(X_i)=Y_i$ for all $i$, then we call $\widetilde{T}$ an \emph{extension} of the map $T$.

Since a completely positive extension may not exist, we ask here for an approximation $\widetilde{T}:\cM_d\ra\cM_{d'}$ of $T$ that is as good as possible in some sense. Our main result in this direction is Theorem \ref{th:cpisapprox}, which exactly characterizes the maps $T$ that can be approximated arbitrarily well as the restriction of a map that is completely positive on the full algebra.

In order to assess the goodness of the approximation, we need to quantify it. This can be done in different ways (see Remark \ref{remark:onthechannelgoodnessmeasure}), but for now we define this measure to be the functional
\begin{equation}\label{eq:distancemeasure}\Delta(\widetilde{T})\,:=\,\sum_{i=1}^N||\widetilde{T}(X_i)-Y_i||_1~,\ee
which quantifies how close a given linear map $\widetilde{T}:\cM_d\ra\cM_{d'}$ is to $T:X_i\mapsto Y_i$. In particular, $\Delta(\widetilde{T})=0$ iff $\widetilde{T}$ is an extension of $T$. So,
\begin{equation}\label{eq:delta}\Delta\,:=\,\inf\big\{\Delta(\widetilde{T})\,\big|\,\widetilde{T}:\cM_d\ra\cM_{d'}~\text{completely~positive}\big\}\end{equation}
quantifies how well the linear map $T:X_i\mapsto Y_i$ can be approximated by a completely positive map defined on the full algebra $\cM_d$. In particular, $\Delta=0$ iff there exists a sequence of completely positive maps $T_n:\cM_d\ra\cM_{d'}$ which converges to $T$ in the sense that $T_n(X_i)\ra Y_i$ as $n\to\infty$ for all $i$.
We will later see that this approximation property is \emph{not} equivalent to the existence of a completely positive extension $\widetilde{T}:\cM_d\ra\cM_{d'}$ of $T$.

Later on, we will uncover several relations between the two questions of {\it(i)} whether a map is completely positive on a subspace $\cS$ and of {\it(ii)} whether it can be approximated arbitrarily well or be extended by a completely positive map defined on the whole space $\cM_d$. The following result already shows that the corresponding optimization problems, namely \eqref{CPSDPminimizingfunction}--\eqref{onepossiblecompactnesscondition} on the one hand and \eqref{eq:delta} on the other, are closely related. In fact, each is a semidefinite program (SDP) and both are related by strong duality \cite{convexoptimization} with $\Delta=-\Gamma$. As this equality will have important consequences in the following, we give a self-contained proof using the separating hyperplane theorem from convex analysis in Appendix \ref{dualityproofapp}.

\begin{proposition}[Strong SDP duality for complete positivity]\label{CPdualityprop}
Let $X_i\in\cM_d$ and $Y_i\in\cM_{d'}$ be Hermitian matrices.
The optimization \eqref{eq:delta}, whose optimal value $\Delta$ quantifies how well the map $X_i\mapsto Y_i$ can be approximated with a completely positive map on $\cM_d$, can be formulated as an SDP.
Up to a minus sign, it is the dual of the SDP (\ref{CPSDPminimizingfunction})--(\ref{onepossiblecompactnesscondition}), whose optimal value $\Gamma$ equals 0 iff the map $X_i\mapsto Y_i$ is completely positive on $\cS:={\rm span}\{X_i\}_i$.
The optimal values of both SDPs are related by
\begin{equation}\Delta = -\Gamma~.\end{equation}
\end{proposition}
\begin{remark}\label{remark:onthechannelgoodnessmeasure}The SDP (\ref{eq:delta}) and the SDP (\ref{CPSDPminimizingfunction})--(\ref{onepossiblecompactnesscondition}) remain dual to each other when instead of (\ref{onepossiblecompactnesscondition}) another semidefinite compactness condition is chosen and at the same time the ``goodness measure'' $\Delta(\widetilde{T})$ in (\ref{eq:distancemeasure}) is modified accordingly.
In particular, note that (\ref{onepossiblecompactnesscondition}) can be written as $\max_i||H_i||\leq1$ and that the summation $\sum_i$ in (\ref{eq:distancemeasure}) is dual to the maximization $\max_i$, whereas the trace-norm is dual to the operator norm in (\ref{onepossiblecompactnesscondition}).
As a further example, Proposition \ref{CPdualityprop} still remains true if instead of (\ref{eq:distancemeasure}) one defines $\Delta(\widetilde{T}):=\max_i||\widetilde{T}(X_i)-Y_i||$ and instead of (\ref{onepossiblecompactnesscondition}) one imposes the constraint $\sum_i||H_i||_1\leq1$.
\end{remark}

\begin{proof}
The trace-norm of a Hermitian matrix $X$ equals the sum of the traces of its positive and negative parts and can thus be expressed as a minimization via
\begin{equation*}
||X||_1=\inf\{\tr{P+Q}\,\big|\,P,Q\geq0,P-Q=X\} \, .
\end{equation*}
Using this in Eq.~\eqref{eq:distancemeasure} and replacing the infimum over all completely positive maps $\widetilde{T}$ in \eqref{eq:delta} by an infimum over all positive Choi matrices $C\in\cM_{d'}\otimes\cM_d$ according to the Choi-Jamiolkowski correspondence (\ref{choicorrespondenceeqn}), one can write the optimization (\ref{eq:delta}) as an SDP:
\bea\Delta~=&\text{infimum~of}~&\sum_{i=1}^N\tr{P_i+Q_i}~,\label{deltasdpinproof1}\\
&\text{subject~to}~&C,\,P_i,\,Q_i\,\geq\,0~,\label{deltasdpinproof2}\\
&&P_i-Q_i\,=\,{\rm tr}_2[C(\ii\otimes X_i^T)]-Y_i\quad\forall i~.\label{deltasdpinproof3}\eea

It follows by straightforward computation that, up to a minus sign, this is dual to the SDP (\ref{CPSDPminimizingfunction})--(\ref{onepossiblecompactnesscondition}); see for example \cite{convexoptimization}. Hereby, in (\ref{deltasdpinproof2}) the constraint $C\geq0$ to completely positive maps follows as the Lagrangian multiplier of the positivity condition (\ref{CPSDPconditions}). As the SDP (\ref{deltasdpinproof1})--(\ref{deltasdpinproof3}) is \emph{strictly feasible} (e.g.~by choosing any $C>0$ and $P_i,Q_i>0$ accordingly), Slater's constraint qualification immediately gives strong duality \cite{rockafellar,convexoptimization}, i.e.~$\Delta=-\Gamma$. An elementary proof of the latter fact is given in Appendix \ref{dualityproofapp}.
\end{proof}

Roughly speaking, Proposition \ref{CPdualityprop} quantitatively relates the degree $\Gamma$, to which a given map $T:X_i\mapsto Y_i$ fails to be completely positive, to the goodness $\Delta=-\Gamma$, with which it can be approximated by a completely positive map on $\cM_d$. In the extremal and qualitative case $\Gamma=\Delta=0$ one obtains the following theorem, which is our main result:
\begin{theorem}[Complete positivity and approximability]\label{th:cpisapprox}
Let $T:\cS\to\cM_{d'}$ be a Hermiticity-preserving linear map on a Hermitian subspace $\cS\subseteq\cM_{d}$.
The following statements are equivalent:
\begin{itemize}
\item[(i)]$T$ is completely positive.
\item[(ii)] For every $\epsilon>0$, there exists a completely positive map $T_\epsilon:\cM_d\ra\cM_{d'}$ such that $||T_\epsilon(X)-T(X)||\leq\epsilon||X||$ for all $X\in\cS$.
\end{itemize}
\end{theorem}

\begin{proof}
Choose a basis $\{B_i\}_{i=1}^N$ of $\cS$, where $N=\dim\cS$.
As $\cS$ is Hermitian, the Hermitian parts $(B_i+B_i^\dagger)/2$ and the anti-Hermitian parts $(B_i-B_i^\dagger)/2i$ belong to $\cS$, so together they form a set of Hermitian matrices spanning the subspace $\cS$.
From this set one can choose a basis $\{X_i\}_{i=1}^N$ of $\cS$ consisting of Hermitian matrices.
Define the images $Y_i:=T(X_i)$.

Assume {\it(i)}. By Theorem \ref{th:CPnessviaSDPtheorem} the infimum in \eqref{CPSDPminimizingfunction} under the conditions (\ref{CPSDPconditions})--(\ref{onepossiblecompactnesscondition}) will be $\Gamma=0$.
But by Proposition \ref{CPdualityprop}, the optimum in \eqref{eq:delta} is then also $\Delta=0$, which means that for any $\epsilon'>0$ one can find a completely positive map $\widetilde{T}:\cM_d\ra\cM_{d'}$ such that $\no{\widetilde{T}(X_i)-T(X_i)}\leq\epsilon'$ for all $i=1,\ldots,N$.
Here we have used the fact that all norms on the finite-dimensional vector space $\cM_{d'}$ are equivalent to the trace-norm in Eq.~(\ref{eq:distancemeasure}).
As the $X_i$ form a basis of $\cS$, we can define a norm $||\cdot||_\cS$ on the vector space $\cS$ by $\no{X}_\cS:=\sum_i|c_i|$ where the $c_i$'s are determined uniquely as the coefficients in the linear combination $X=\sum_i c_i X_i$.
Again, as $\cS$ is finite-dimensional, there exists a constant $s>0$ such that $\no{X}_\cS\leq s\no{X}$. So, for any $X\in\cS$ we can write
\be\no{\widetilde{T}(X)-T(X)}\,=\,\big|\big|(\widetilde{T}-T)\big(\sum_ic_iX_i\big)\big|\big|\,\leq\,\sum_i|c_i|\, \no{\widetilde{T}(X_i)-T(X_i)}\,\leq\,\epsilon' \no{X}_\cS\,\leq\,s\epsilon' \no{X}~,\nonumber\ee
so that setting $\epsilon':=\epsilon/s$ above and $T_\epsilon:=\widetilde{T}$ gives the desired map.

Conversely, if {\it(ii)} holds then in particular there exist completely positive maps $T_\epsilon:\cM_d\ra\cM_{d'}$ which approximate the map $T$ arbitrarily well on the particular inputs $X_i$. Again using that all norms on $\cM_{d'}$ are equivalent and that the sum in (\ref{eq:distancemeasure}) is over a fixed number $N$ of terms, one sees that $\Delta(T_\epsilon)$ can be made arbitrarily close to $0$ with completely positive maps $T_\epsilon$, so that $\Delta=0$ in (\ref{eq:delta}). Proposition \ref{CPdualityprop} gives that the infimum in (\ref{CPSDPminimizingfunction}) is $0$ under the conditions (\ref{CPSDPconditions}) and (\ref{onepossiblecompactnesscondition}). Due to homogeneity, the infimum in (\ref{CPSDPminimizingfunction}) remains $0$ even if the latter condition is dropped, so that Theorem \ref{th:CPnessviaSDPtheorem} becomes applicable to finally prove {\it(i)}.
\end{proof}

\begin{remark}\label{remark:aftercpisapproxtheorem}
More generally, if {\it(ii)} in Theorem \ref{th:cpisapprox} is satisfied for any (not necessarily Hermiticity-preserving) linear map $T:\cS\ra\cM_{d'}$ on a (not necessarily Hermitian) subspace $\cS\subseteq\cM_d$, then {\it(i)} holds. However, simple examples show that dropping any of these two presuppositions will generally invalidate the implication {\it(i)}\,$\Rightarrow$\,{\it(ii)}.

To prove the first statement of the previous paragraph, note that $T$ is the pointwise limit on $\cS$ of a sequence of linear Hermiticity-preserving maps $T_\epsilon$ that are defined on the whole space $\cM_d$.
It is then easy to see that $T$ can be extended (in a unique way) to a linear and Hermiticity-preserving map $T'$ on the Hermitian subspace $\cS':={\rm span}\{\cS,\cS^\dagger\}\subseteq\cM_d$, and that then condition {\it(ii)} will hold for $T'$.
What we have proved above (Theorem \ref{th:cpisapprox}) shows that $T'$ is then completely positive on $\cS'$.
Thus, the restriction of $T'$ to $\cS$, which is just $T$, is completely positive as well, as is immediate from the definition of complete positivity (Definition \ref{def:CPmaponS}).
\end{remark}

Contrary to what one might expect, if $T$ is completely positive on a Hermitian subspace (so that condition {\it(i)} of Theorem \ref{th:cpisapprox} holds), then there does not necessarily exist a completely positive extension $\widetilde{T}$ to the whole space, i.e.~{\it(ii)} of Theorem \ref{th:cpisapprox} may not hold for $\epsilon=0$.
This is demonstrated by the following example.

\begin{example}[No completely positive extension]\label{ex:noexactcpextension}
Define $\cS:={\rm span}\{\Psi_0:=\ket{0}\bra{0},\sigma_x\}\subset\cM_2$ and a linear map $T:\cS\ra\cM_2$ by $T(\Psi_0):=\Psi_0$ and $T(\sigma_x):=\sigma_z$ (where $\sigma_i$ are the Pauli sigma matrices). According to Example \ref{ex:easytocheckcp}, $T$ is completely positive, and it is also Hermiticity-preserving on the Hermitian subspace $\cS$.
In particular, it satisfies the conditions of Theorem \ref{th:cpisapprox}.
Note that, furthermore, $T$ is trace-preserving.

We claim that $T$ has no completely positive extension to $\cM_2$. This can be seen in the following way. Assume that a positive extension $\widetilde{T}$ exists, in particular for $\Psi_1:=\ket{1}\bra{1}$, we have that $\widetilde{T}(\Psi_1)=:\left(\begin{matrix}a&b\\ b^*&c\end{matrix}\right)\geq0$ and hence $a,c\geq0$.
Observe that $X_\delta:=\Psi_0+\delta\sigma_x+\delta^2\Psi_1=\left(\begin{matrix}1&\delta\\\delta&\delta^2\end{matrix}\right)$ is positive for all $\delta\in\R$, but $\widetilde{T}(X_\delta)=\left(\begin{matrix}1+\delta+\delta^2a&\delta^2b\\\delta^2b^*&-\delta+\delta^2c\end{matrix}\right)$ has negative determinant for small $\delta$,
\be\det\left[\widetilde{T}(X_\delta)\right]\,=\,\delta\left[(\delta c-1)(1+\delta+\delta^2a)-\delta^3|b|^2\right]\,<\,0\qquad\text{for}~\,\delta\in(0,1/c)~,\nonumber
\ee
and so $\widetilde{T}(X_\delta)$ cannot be positive for $\delta\in(0,1/c)$, contradicting the assumption.
\end{example}

The next example illustrates the approximations by completely positive maps of Theorem \ref{th:cpisapprox}.

\begin{example}[Approximation by completely positive maps]\label{approxbycpmapsexample}
As $T$ from the previous example is completely positive on $\cS$, we can approximate it arbitrarily closely by completely positive maps on $\cM_2$ in the sense of Theorem \ref{th:cpisapprox}.
For example, for $\epsilon>0$ one can set
\be T_\epsilon(X)\,:=\,K_1XK_1^\dagger+K_2AK_2^\dagger~~\quad\text{with}\quad~
K_1\,:=\,\left(\begin{matrix}1&1/2\\0&0\end{matrix}\right),~
K_2\,:=\,\left(\begin{matrix}0&0\\\epsilon&-1/2\epsilon\end{matrix}\right)~.\nonumber\ee
Obviously, $T_\epsilon$ is completely positive as it is given in Kraus representation.
Its action on the basis elements of $\cS={\rm span}\{\Psi_0,\sigma_x\}$ is $T_\epsilon(\Psi_0)=\Psi_0+\epsilon^2\Psi_1$ and $T_\epsilon(\sigma_x)=\sigma_z$, so that on $\cS$ it obviously approximates $T$ as $\epsilon\to0$.

The norm of the Kraus operator $K_2$ above diverges as the approximation becomes better and better. This observation can be rephrased without reference to a particular Kraus decomposition by saying that
\be T_\epsilon^*(\ii)\,=\,K_1^\dagger K_1+K_2^\dagger K_2\,=\,\left(\begin{matrix}1+\epsilon^2&0\\0&(1+1/\epsilon^2)/4\end{matrix}\right)\nonumber\ee
is unbounded in the limit of arbitrarily good approximation ($\epsilon\to0$). 
We now prove that this phenomenon necessarily occurs in the case where, as in the present example, no completely positive extension exists but where an approximating sequence exists.
\end{example}

\begin{theorem}[Unboundedness for ``true'' approximations]\label{unboundednesstrueapprox}
Let $T:\cS\ra\cM_{d'}$ be a map on a subspace $\cS\subset\cM_d$.
Assume that $T_k:\cM_d\ra\cM_{d'}$ is a sequence of completely positive maps that converges pointwise on $\cS$ to $T$, but that $T$ does not possess a completely positive extension $\widetilde{T}:\cM_d\ra\cM_{d'}$.
Then the sequence $\{T_k\}$ is unbounded in the sense that
\be\no{T_k^*(\ii)}\,\to\,\infty\quad\text{as}~~k\to\infty~.\label{normsofTstartoinfty}\ee
\end{theorem}

\begin{proof}
Assume, contrary to (\ref{normsofTstartoinfty}), that there is a real constant $B>0$ such that for any $K\in\nat$ there exists $k>K$ with $\no{T_k^*(\ii)}\leq B$. Then there is a subsequence $T_{k(m)}$, which we will simply call $T_m$, that is bounded in the sense that $\no{T_m^*(\ii)}\leq B$ for all $m\in\nat$.
As $T_m^*(\ii)\geq0$, this implies $\tr{T_m^*(\ii)}\leq Bd$ for all $m\in\nat$.

As in the proof of Remark \ref{remark:aftercpisapproxtheorem}, $T_m$ converges to a linear Hermiticity-preserving map $T'$ on the Hermitian subspace $\cS':={\rm span}\{\cS,\cS^\dagger\}$.
Again, we consider a basis of $\cS'$ consisting of Hermitian matrices $X_i$ and define $Y_i:=T'(X_i)$. Then the infimum in (\ref{eq:delta}) is $\Delta=0$ since $T'$ was constructed as the pointwise limit of a sequence of completely positive maps $T_m$, so in particular $T_m(X_i)\to Y_i$ as $m\to\infty$ for all $i=1,\ldots,\dim\cS'$.
As the $T_m$ are uniformly bounded in the above sense, we even have
\be0\,=\,\inf\Big\{\sum_i||\widetilde{T}(X_i)-T'(X_i)||_1\,\,\Big|\,\,\widetilde{T}:\cM_d\ra\cM_{d'}~\text{completely~positive},\,\tr{\widetilde{T}^*(\ii)}\leq Bd\Big\}~.\nonumber\ee
Now, the last infimum is attained since a continuous function is optimized over the compact set of completely positive maps satisfying $\tr{T^*(\ii)}\leq Bd$ (this corresponds to supplementing the SDP (\ref{deltasdpinproof1})--(\ref{deltasdpinproof3}) by the constraint $\tr{C}\leq Bd$, which together with the constraint $C\geq0$ in (\ref{deltasdpinproof2}) makes the set of feasible Choi operators $C$ compact). But this attaining of the infimum implies that $T'$ has a completely positive extension $\widetilde{T}:\cM_d\ra\cM_{d'}$, which however contradicts the assumptions of the theorem as $\widetilde{T}$ would also be a completely positive extension of $T$.
\end{proof}

A way to explicitly find an approximating sequence of completely positive maps is to solve the SDP (\ref{deltasdpinproof1})--(\ref{deltasdpinproof3}) algorithmically, e.g.~via interior point methods \cite{convexoptimization}. This gives a sequence of Choi matrices $C_k$ such that the objective function converges to its optimum $\Delta$, and these can then be turned into completely positive linear maps $T_k$ via the Choi-Jamiolkowski isomorphism (\ref{choicorrespondenceeqn}).

In Section \ref{CPextensionssect} we will discuss various situations and conditions which guarantee for a given map $T$ the existence of a completely positive extension, and not merely of an approximation as in Theorem \ref{th:cpisapprox} {\it(ii)} or Example \ref{approxbycpmapsexample}. As in the present section, several of these analytic results can be derived from the SDP formulation of the problem.
And besides providing an analytic approach, the SDP formulation further allows for efficient and certifiable numerical algorithms \cite{convexoptimization} to decide the existence of an extension and to find optimal approximations; we will illustrate this numerical use in subsection \ref{counterexpsection} to disprove a conjecture.

\section{Completely positive extensions}\label{CPextensionssect}

We have seen in Theorem \ref{th:cpisapprox} that, in a Hermitian setting, when the linear map $T$ is completely positive on a subspace, it can be approximated arbitrarily well with completely positive maps defined on the full algebra. In some cases, complete positivity on a subspace implies a stronger result that there actually exists a completely positive extension to the full algebra. In this section we discuss several conditions leading to this conclusion.

\subsection{Subspace containing a strictly positive operator}

If, beyond the assumptions in Theorem \ref{th:cpisapprox} {\it(i)}, $T$ is a completely positive map on an operator system $\cS$ (i.e.~a Hermitian subspace $\cS\subseteq\cM_d$ containing the identity matrix $\ii$), then it is well-known that $T$ has a completely positive extension $\widetilde{T}$ defined on all of $\cM_d$ (see e.g.~\cite{paulsen}, chapter 6).
The following theorem generalizes this fact to the case where $\cS$ may not be an operator system but contains a strictly positive operator, i.e.~a positive matrix with full rank.
In subsection \ref{sec:spannedbydensityops} we will generalize this theorem to the case relevant in quantum theory, where the subspace $\cS$ is spanned by (not necessarily full-rank) density operators.

For the following theorem we provide three different albeit related proofs. On the one hand these should clarify the relation to the existing results (e.g.~\cite{paulsen}) and emphasize the underlying mechanism, and on the other hand should make the connection to the more quantitative formulation of complete positivity following Theorem \ref{th:CPnessviaSDPtheorem}, including the formulation as an SDP in Proposition \ref{CPdualityprop} and the result on approximative extensions in Theorem \ref{th:cpisapprox} (see esp.~third proof below).

\begin{theorem}[Completely positive extension;~Arveson's extension theorem]\label{th:exactcpextension}
Let $\cS\subseteq\cM_d$ be a Hermitian subspace which contains a strictly positive operator $P$.
Let $T:\cS\ra\cM_{d'}$ be a completely positive map. Then there exists a completely positive map $\widetilde{T}:\cM_d\ra\cM_{d'}$ extending $T$.
\end{theorem}
\begin{proof}[First proof of Theorem \ref{th:exactcpextension} (using the Hahn-Banach theorem)]With the stated assumptions, $T$ is Hermiticity-preserving: if $X\in\cS$ is Hermitian, then there exists $\epsilon>0$ such that $P+\epsilon X\geq0$.
So, $T(P+\epsilon X)=T(P)+\epsilon T(X)$ is also positive and in particular Hermitian, and as $T(P)$ is Hermitian, $T(X)$ will be as well.
In the following we can thus first focus on extensions from the real vector space $\cS_h:=\cS\cap S_d$ to the space $S_d$ of Hermitian matrices in $\cM_d$. Later on, the action of $T$ on $S_d$ will determine its action on all of $\cM_d$.

In order to use the Hahn-Banach theorem, we need a sublinear functional. For this, fix a strictly positive operator $Q\in\cM_{d'}$ (e.g.~$Q=\ii$), and define the following functional for $B\in S_d\otimes S_{d'}$ (cf.~\cite{hilbertmetricpaper}):
\be\sup(B/P\otimes Q)\,:=\,\inf\big\{\lambda\geq0\,\big|\,\lambda\,P\otimes Q\geq B\big\}~.\nonumber\ee
It is easy to see that this functional is sublinear in $B$, i.e.~positive homogeneous and subadditive.
In particular, the infimum in the last equation is always finite as $P\otimes Q$ is strictly positive.

Define also a linear functional on the subspace $\cS_h\otimes S_{d'}\subseteq S_d\otimes S_{d'}$:
\be\tau(B)\,:=\,\bra{\Omega_{d'}}\,(T\otimes{\rm id})(B)\,\ket{\Omega_{d'}}\label{definelinearfunctionalonsubspace}~,\ee
where $\ket{\Omega_{d'}}=\sum_{i=1}^{d'}\ket{ii}$ is the unnormalized maximally entangled state. For linear maps $T$ not defined on the whole matrix algebra, this functional is analogous in spirit to the Choi operator, cf.~Eq.~(\ref{choicorrespondenceeqn}).

By definition, we have $\sup(B/P\otimes Q)\,P\otimes Q\,-\,B\geq0$ for all $B\in S_d\otimes S_{d'}$. Employing this inequality for any $B\in\cS_h\otimes S_{d'}$, exploiting complete positivity of $T$, and using the linearity of $\tau$ in (\ref{definelinearfunctionalonsubspace}), we obtain:
\be\tau(B)\,\leq\,\tau(P\otimes Q)\,\sup(B/P\otimes Q)\qquad\forall B\in\cS_h\otimes S_{d'}~.\nonumber\ee

With this domination of the linear functional by a sublinear functional as the last ingredient, we can apply the Hahn-Banach theorem, see e.g.~\cite{conwayfunctana}, to get the existence of a linear functional $\sigma:S_d\otimes S_{d'}\ra\R$ which extends $\tau$, meaning $\sigma(B)=\tau(B)$ for all $B\in\cS_h\otimes S_{d'}$, and which satisfies the domination property on all of $S_d\otimes S_{d'}$:
\be\sigma(B)\,\leq\,\sigma(P\otimes Q)\,\sup(B/P\otimes Q)\qquad\forall B\in S_d\otimes S_{d'}~.\nonumber\ee
Via $\C$-linearity, $\sigma$ can be extended to a linear functional on the complex vector space $\cM_d\otimes\cM_{d'}$, and through Riesz' representation theorem it can be written as $\sigma(B)=\tr{CB}$ for some $C\in\cM_d\otimes\cM_{d'}$. If $R\in\cM_d\otimes\cM_{d'}$ is positive, then plugging $B:=-R$ into the last equation and noting $\sup(-R/P\otimes Q)=0$ gives $0\leq\sigma(R)=\tr{CR}$ for all $R\geq0$. Thus, $C\geq0$.

Define now, in close analogy with Eq.~(\ref{choicorrespondenceeqn}), the linear map $\widetilde{T}:\cM_d\ra\cM_{d'}$:
\be\widetilde{T}(X)\,:=\,{\rm tr}_1[C(X\otimes\ii)]^T\,=\,{\rm tr}_1[C^T(X^T\otimes\ii)]~.\nonumber\ee
Then $\widetilde{T}$ extends $T$ since, for all $X\in\cS$, the inner product with all $D\in\cM_{d'}$ agrees:
\bea\tr{\widetilde{T}(X)\,D}&=&{\rm tr}_2\big[{\rm tr}_1[C(X\otimes\ii)]^T\,D\big]\,=\,\tr{C(X\otimes D^T)}\,=\,\sigma(X\otimes D^T)\nonumber\\
&=&\tau(X\otimes D^T)\,=\,\bra{\Omega_{d'}}\,T(X)\otimes D^T\,\ket{\Omega_{d'}}\,=\,\tr{T(X)\,D}~.\nonumber\eea
Also, $\widetilde{T}$ is completely positive as $C^T$ in its defining equation is positive \cite{Choi75}.
\end{proof}

We present a second proof which uses a well-known extension result from operator algebra theory:
\begin{proof}[Second proof of Theorem \ref{th:exactcpextension} (using results on operator systems)]Define
\be\cS'\,:=\,P^{-1/2}\,\cS\,P^{-1/2}\,=\,\big\{P^{-1/2}XP^{-1/2}\,\big|\,X\in\cS\big\}\,\subseteq\,\cM_d~,\nonumber\ee
which is clearly an operator system.
The map
\be
T':\cS'\ra\cM_{d'},~T'(X)\,:=\,T(P^{1/2}\,X\,P^{1/2})\nonumber
\ee
is linear and completely positive as $T$ is completely positive.
So, $T'$ is a completely positive map on an operator system $\cS'$ and therefore we can apply the usual extension theorem from the theory of operator algebras (e.g.~\cite{paulsen}, Theorem 6.2) to get the existence of a completely positive extension $T'':\cM_d\ra\cM_{d'}$ of $T'$. Then, defining $\widetilde{T}(X):=T''(P^{-1/2}XP^{-1/2})$ for $X\in\cM_d$, we have a completely positive map that extends $T$ since, for all $X\in\cS$,
\begin{equation*}
\widetilde{T}(X)=T''(P^{-1/2}XP^{-1/2})=T'(P^{-1/2}XP^{-1/2})=T(X) \, .
\end{equation*}
\end{proof}

Note that this second proof is just a reformulation of the first one, as implicitly it also uses the Hahn-Banach theorem to guarantee the existence of the extension $T''$.
We included this proof here as it uses a well-known extension result from operator algebra theory, and in order to show how naturally and easily the assumption of an operator system can be replaced by the assumption of a Hermitian subspace containing a full-rank state, which is a more relevant case for quantum theory.

The third proof is phrased in the language of semidefinite programming and convex analysis, which guided our approach to relate the notions of complete positivity and of approximability by completely positive maps in Sections \ref{CPnesssect} and \ref{sectapproxextensions}. We exhibit this proof here also in order to once explain the detailed reasoning via strong SDP duality that is useful in the subsequent treatment.

\begin{proof}[Third proof of Theorem \ref{th:exactcpextension} (using Slater's constraint qualification for SDPs)]
Like at the beginning of the proof of Theorem \ref{th:cpisapprox}, we can find a Hermitian basis $\{X_i\}_i$ of $\cS$ and set $Y_i:=T(X_i)$, and now we can even choose $X_1=P>0$.
This implies that the SDP (\ref{CPSDPminimizingfunction})--(\ref{onepossiblecompactnesscondition}), which decides whether $T$ is completely positive, enjoys the property of being \emph{strictly feasible} \cite{rockafellar,convexoptimization}, meaning that there exists an assignment for the variables $H_i$ such that the constraints (\ref{CPSDPconditions})--(\ref{onepossiblecompactnesscondition}) are satisfied with strict inequality. For example, setting $H_1=\epsilon\ii$ for small enough $\epsilon>0$ and $H_i=0$ for $i\geq2$ suffices.

Slater's constraint qualification from convex analysis \cite{rockafellar,convexoptimization} then ensures that the optimum of the dual SDP (\ref{eq:delta}) is attained (and that the optima of both SDPs satisfy $\Delta=-\Gamma$, which we already know from Proposition \ref{CPdualityprop}). In an elementary way, this can be proven with the separating hyperplane theorem, analogously to the reasoning of optimal variables $\hat{H}_i$ in Appendix \ref{dualityproofapp}.

As $T$ is completely positive by assumption, we have $\Delta=-\Gamma=0$, so that the attaining of the infimum in (\ref{eq:delta}) guarantees the existence of a completely positive map $\widetilde{T}:\cM_d\ra\cM_{d'}$ satisfying $\widetilde{T}(X_i)=Y_i=T(X_i)$ for all $i$, i.e.~$\widetilde{T}$ extends $T$.
This proves the theorem.
\end{proof}
The latter proof actually establishes a more general result about the optimizations from Proposition \ref{CPdualityprop}, namely that the infimum in Eq.~(\ref{eq:delta}) is attained whenever $\cS=\cS^\dagger$ contains a strictly positive operator. In other words, in this case there exist ``best'' completely positive approximations.

\begin{remark}\label{dualityinoppositedirection}The traditional application of the Hahn-Banach theorem directly to the linear functional $\tau$ (\ref{definelinearfunctionalonsubspace}) on $\cS_h\otimes\cS_{d'}$ (as in the first proof of Theorem \ref{th:exactcpextension}) \emph{cannot} yield the approximation result in Theorem \ref{th:cpisapprox}. For this, instead, the Hahn-Banach theorem has to be applied to a functional on the space given by the \emph{dual} problem, like the separating hyperplane theorem in the proof of Proposition \ref{CPdualityprop} (see Appendix \ref{dualityproofapp}) was applied to the set $\cP$ defined by the constraints of the \emph{dual} problem. In this sense, the approximation result in Theorem \ref{th:cpisapprox} seems to improve conceptually on the usual completely positive extension results \cite{paulsen}.
\end{remark}

\subsection{Subspace supported by density operators}\label{sec:spannedbydensityops}

In a quantum setting, the map $T$ might be given by its action on a set of quantum states.
For example, this is the case in the situation described in the Introduction: Suppose that for a set of input quantum states $\rho_i$, their respective outcome states $\rho'_i$ after a time-evolution (possibly a black box) have been determined through state tomography.
This fixes a map $T:\rho_i\mapsto\rho'_i$ on the subset $\cS={\rm span}\{\rho_i\}_i$ in the Schr\"odinger picture.
One can now ask whether this map can possibly originate from a completely positive evolution on the whole system, the minimum requirement for a memoryless quantum mechanical time-evolution.
The situation that $\cS$ is spanned by positive operators arises also in the setting where, in the Heisenberg picture, the evolution of a set of POVM elements has been determined.

Note that in this case $\cS$ might still not contain a positive full-rank operator, which is required to apply Theorem \ref{th:exactcpextension}, and even less so the identity $\ii$. However, a completely positive extension exists even in this situation if only $T$ was completely positive.
We remark that this is exactly the result covered by Theorem 3 in \cite{jencovacp}.
Further, we reiterate that if merely the image $\rho'_i$ of each input state $\rho_i$ is known, then one can apply the SDP (\ref{CPSDPminimizingfunction})--(\ref{onepossiblecompactnesscondition}) with $X_i:=\rho_i$, $Y_i:=\rho'_i$ to check this map for both linearity and complete positivity, returning $\Gamma<0$ iff either requirement is violated.

\begin{corollary}[Subspace spanned by positive operators; Theorem 3 in \cite{jencovacp}]\label{corollaryspannedbydensitymatrices}
Let $\cS={\rm span}\{P_i\}_i$ be spanned by positive operators $P_i\in\cM_d$ ($i=1,\ldots,N$), and assume that $P_i\mapsto P'_i\in\cM_{d'}$ defines a completely positive map $T:\cS\ra\cM_{d'}$.
Then there exists a completely positive map $\widetilde{T}:\cM_d\ra\cM_{d'}$ extending $T$.
\end{corollary}

\begin{proof}
First note that $\cS\subseteq\cM_d$ is a Hermitian subspace, and that the map $T$ is Hermiticity-preserving as $P'_i=T(P_i)\geq0$ and the $P'_i$ are thus in particular Hermitian. This ensures applicability of the SDP (\ref{CPSDPminimizingfunction})--(\ref{onepossiblecompactnesscondition}).

Define $P\in\cM_d$ to be the orthogonal projection onto the support of $\sum_{i=1}^NP_i$.
Then $\cS\subseteq P\cM_dP=:\cM'$. As $\cM'$ is a matrix algebra and $\sum_iP_i\in\cS$ has full-rank in $\cM'$, Theorem \ref{th:exactcpextension} guarantees the existence of a completely positive extension $T':\cM'\ra\cM_{d'}$ of $T$.
Then $\widetilde{T}:\cM_d\ra\cM_{d'}$ can be defined by $\widetilde{T}(X):=T'(PXP)$ for all $X\in\cM_d$.
\end{proof}

Note that Theorem \ref{th:exactcpextension} is subsumed by Corollary \ref{corollaryspannedbydensitymatrices}, since if $P$ is a strictly positive element of $\cS$ and if the Hermitian matrices $\{X_i\}_i$ span $\cS$, then for small enough $\epsilon>0$ the set $\{P\}\cup\{P_i:=P-\epsilon X_i\}_i$ consists of positive operators spanning $\cS$.

\subsection{Subspace containing no positive operator}\label{subspacewopositiveoperators}

Here we present another situation in which a completely positive extension exists and which is not covered by the previous ideas.
As we have seen in Proposition \ref{prop:oposelementthencp}, on a space $\cS\subset\cM_d$ not containing any nonzero positive elements any linear map $T:\cS\ra\cM_{d'}$ is completely positive.
We will prove here that then (and under natural Hermiticity assumptions) $T$ possesses a completely positive extension $\widetilde{T}$ to the full algebra $\cM_d$, i.e.~beyond the mere approximation guaranteed by Theorem \ref{th:cpisapprox} {\it(ii)}.
An example of this situation is the case where $\cS$ is spanned by spin operators whose time-evolution in the Heisenberg picture is known, cf.~Proposition \ref{prop:oposelementthencp}.
The following proof combines the characterization of complete positivity via Theorem \ref{th:CPnessviaSDPtheorem} with the extension result from Theorem \ref{th:exactcpextension}.

\begin{proposition}[Completely positive extensions on subspaces without positive elements]\label{cpextensionwoposelement}Let $\cS\subset\cM_d$ be a Hermitian subspace containing no positive operator except $0$, and let $T:\cS\ra\cM_{d'}$ be a Hermiticity-preserving linear map. Then there exists a completely positive linear map $\widetilde{T}:\cM_d\ra\cM_{d'}$ extending $T$.
\end{proposition}

\begin{proof}On $\cS':={\rm span}\{\cS,\ii\}$ we define a linear map $T'$ extending $T$ and satisfying $T'(\ii)=c\ii$, and show that it is completely positive for some $c>0$.
 As $T'$ satisfies the conditions of Theorem \ref{th:exactcpextension}, there exists an extension $\widetilde{T}$ of $T'$, and thus of $T$, as desired.

As before, let $\{X_i\}_i$ be a basis of $\cS$ consisting of Hermitian matrices. Now we want to apply Theorem \ref{th:CPnessviaSDPtheorem} to show that $T'$ is completely positive for large enough $c>0$. Thus, let $H_0,H_i\in\cM_{d'}$ such that
\be\ii\otimes H_0\,+\,\sum_i X_i\otimes H_i\,\geq0~.\nonumber\ee
As the $X_i$ and $\ii$ are linearly independent, all $H_0,H_i$ have to be Hermitian.
If $\bra{\psi}H_0\ket{\psi}<0$ then the last inequality would give
\begin{equation*}
\sum_i\bra{\psi}H_i\ket{\psi}X_i\geq-\bra{\psi}H_0\ket{\psi}\ii>0 \, ,
\end{equation*}
implying the existence of a nonzero positive operator in $\cS$.
Thus, $H_0$ is positive.

If not all $H_i$ are zero then $P:=\sum_i X_i\otimes H_i$ is nonzero since the $X_i$ are linearly independent.
And in fact, $P$ has a negative eigenvalue, since in the contrary case that $P$ would be positive, its partial trace
\begin{equation*}
Q:={\rm tr}_2[P]=\sum_i\tr{H_i}X_i
\end{equation*}
 would be positive, which however is an element of $\cS$ and thus would have to be $0$, such that $P$ would have to be zero due to $\tr{P}={\rm tr}_1[Q]=0$.
This however contradicts the previous observation that $P\neq0$.
So the quantity
\be\lambda\,:=\,\sup\Big\{\lambda_{min}\Big(\sum_i X_i\otimes H_i\Big)\,\Big|\,H_i=H_i^\dagger,\,\sum_i||H_i||=1\Big\}\,<\,0\nonumber\ee
is strictly smaller than $0$ as it is the supremum of a continuous and negative function over a compact set and is thus attained ($\lambda_{min}(\cdot)$ denotes the smallest eigenvalue of a Hermitian matrix).

So, by positivity of $H_0$,  by the first inequality above, and by positive homogeneity of $\lambda_{min}(\cdot)$, we get
\be\tr{H_0}\,\geq\,||H_0||\,=\,||\ii\otimes H_0||\,\geq\,-\lambda_{min}\Big(\sum_iX_i\otimes H_i\Big)\,\geq\,-\lambda\sum_i||H_i||~.\nonumber\ee
Now we are ready to verify the condition of Theorem \ref{th:CPnessviaSDPtheorem} for $T'$:
\bea
\tr{T'(\ii)^TH_0}\,+\,\sum_i\tr{T'(X_i)^TH_i}&=&c\,\tr{H_0}\,+\,\sum_i\tr{T(X_i)^TH_i}\nonumber\\
&\geq&-\lambda c\,\sum_i||H_i||\,-\,\sum_i||T(X_i)||_2\,||H_i||_2\nonumber\\
&\geq&\big(-\lambda c-\sqrt{d'}\max_j||T(X_j)||_2\big)\,\sum_i||H_i||~\,\geq~\,0\nonumber\eea
for any chosen $c\geq\sqrt{d'}\max_j||T(X_j)||_2/(-\lambda)$.
\end{proof}

From the point of view of convex analysis, the situation covered by Proposition \ref{cpextensionwoposelement} is an example where Slater's constraint qualification is not satisfied for the SDP (\ref{CPSDPminimizingfunction})--(\ref{onepossiblecompactnesscondition}) but its dual (\ref{eq:delta}) nevertheless attains its optimum if $\Delta=0$.

\begin{remark}[Freedom in the extension $\widetilde{T}$ in Proposition \ref{cpextensionwoposelement}]Under the assumptions of Proposition \ref{cpextensionwoposelement}, there exists a codimension-1 Hermitian subspace $\cS_1\subset\cM_d$ that includes $\cS\subseteq\cS_1$ but does not contain any positive operator except $0$. We show this in the following paragraph. Thus, one can first choose an \emph{arbitrary} (linear and Hermiticity-preserving) extension of $T$ to $\cS_1$, and then extend this in the remaining direction to a completely positive map $\widetilde{T}$ on $\cM_d$ according to Proposition \ref{cpextensionwoposelement}.

If $\cS$ does not contain any positive element except $0$, then we first prove that its orthogonal complement $\cS^\perp$ contains a strictly positive operator. Assume, to the contrary, that $\cS^\perp$ and $\cC^\circ$ (the interior of the cone $\cC$ of positive elements) are disjoint. Then, by the separating hyperplane theorem for convex sets \cite{rockafellar}, there exists a Hermitian operator $H\neq0$ and a real number $h$ such that
\bea\tr{HX}&\leq&h\quad\forall X=X^\dagger\in\cS^\perp~,\nonumber\\
\tr{HX}&\geq&h\quad\forall X\in\cC^\circ~.\nonumber\eea
As $0$ is in the closure of both sets, we have $h=0$ and thus the second inequality implies that $H$ is an element of the dual cone $\cC^*=\cC$, i.e.~$H$ is positive and nonzero. On the other hand, for any $X=X^\dagger\in\cS^\perp$ we have $(-X)=(-X)^\dagger\in\cS^\perp$, so that by $\tr{HX}\leq0$ and $\tr{H(-X)}\leq0$ one gets $\tr{HX}=0$, i.e.~$H\in(\cS^\perp)^\perp=\cS$. This contradicts the presupposition that $\cS$ does not contain any nonzero positive element, and so our assumption $\cS^\perp\cap\cC^\circ=\emptyset$ was wrong.

Let thus be $K\in\cS^\perp\cap\cC^\circ$, and define the codimension-1 Hermitian subspace $\cS_1:=\{X=X^\dagger\,\big|\,\tr{KX}=0\}=K^\perp$. As $K\in\cS^\perp$, we obviously have $\cS_1\supseteq(\cS^\perp)^\perp=\cS$, and as $K$ is strictly positive we have $\tr{KB}>0$ for any nonzero positive $B$, which are thus not contained in $\cS_1$.
\end{remark}

\begin{example}[Permutations of Pauli spin operators]\label{examplepermutations}As an illustration of Proposition \ref{cpextensionwoposelement}, here we take $\cS:={\rm span}\{\sigma_x,\sigma_y,\sigma_z\}\subset\cM_2$ and first consider the linear map on $\cS$ that permutes the Pauli spin operators cyclically:
\be T^{(123)}:\,\sigma_x\mapsto\sigma_y,~\sigma_y\mapsto\sigma_z,~\sigma_z\mapsto\sigma_x~.\label{actiononspinobservales1example}\ee
According to Proposition \ref{cpextensionwoposelement} this map has a completely positive extension, and in fact one can see that it can be implemented as a unitary conjugation $X\mapsto UXU^\dagger$ on $\cM_2$, with $U=(\ii-i\sigma_x-i\sigma_y-i\sigma_z)/2$.
In particular, this map is unital, and thus a quantum channel in the Heisenberg picture whose action on the spin observables is determined by \eqref{actiononspinobservales1example}.

On the other hand, even though the map
\be T^{(12)}:\,\sigma_x\mapsto\sigma_y,~\sigma_y\mapsto\sigma_x,~\sigma_z\mapsto\sigma_z\label{actiononspinobservales2example}\ee
has a completely positive extension by Proposition \ref{cpextensionwoposelement}, this extension \emph{cannot} be chosen unital, which can be seen in the following way. From the representation $\ket{\Omega}\bra{\Omega}=(\sigma_x\otimes\sigma_x-\sigma_y\otimes\sigma_y+\sigma_z\otimes\sigma_z+\ii\otimes\ii)/2$ of the (unnormalized) maximally entangled state, the Choi matrix of the map $T^{(12)}$ extended by $\ii\mapsto c\ii$ is
\be C^{(12)}\,=\,\big(\sigma_y\otimes\sigma_x-\sigma_x\otimes\sigma_y+\sigma_z\otimes\sigma_z+c\,\ii\otimes\ii\big)/2~,\nonumber\ee
which is positive iff $c\geq3$.
Similar reasoning for the map $T^{(123)}$ above gives that its extension with $\ii\mapsto c\ii$ is completely positive iff $c\geq1$, the case $c=1$ exactly being the unitary conjugation (in this case, the Choi matrix has rank $1$ and the map is unital, so it is unitary).
It is also interesting to note that, even though $T^{(12)}$ does not have a completely positive unital extension, it does have a \emph{positive} unital extension, namely the antiunitary map $X\mapsto VX^TV^\dagger$ with $V={\rm diag}(1,i)$.
Furthermore, if any of the three input-output relations in (\ref{actiononspinobservales2example}) is dropped, then a completely positive extension exists and can chosen to be a unitary conjugation.
\end{example}

\section{Variations and generalizations of the extension problem}\label{variationssect}
In this section we will discuss extension questions that encompass aspects beyond complete positivity, on the one hand imposing further constraints such as trace-preservation, but on the other also loosing constraints such as the requirement of a deterministic quantum operation on the inputs.

We will in this section always require $\cS\subseteq\cM_d$ to be a Hermitian subspace and $T:\cS\ra\cM_{d'}$ a Hermiticity-preserving linear map.

\subsection{Completely positive and trace-preserving extensions}\label{cptpsubsection}

A question of particular relevance for quantum theory is the existence of a completely positive and trace-preserving extension $\widetilde{T}:\cM_d\ra\cM_{d'}$ of a linear map $T$ defined on a subspace $\cS$, i.e.~a quantum channel extension. A map $T$ is trace-preserving iff $T^*(\ii)=\ii$. As this constraint is affine (in $T$), the question of the existence of a completely positive trace-preserving extension can be phrased as a semidefinite program (SDP) as well, akin to Eq.~(\ref{eq:delta}) or Eqs.~(\ref{deltasdpinproof1})--(\ref{deltasdpinproof3}) for the case without the requirement of trace-preservation. There are two obvious ways how to do this: {\it(a)} One could restrict the infimum in (\ref{eq:delta}) to all completely positive \emph{and} trace-preserving linear maps $\widetilde{T}$, or equivalently add the trace-preservation constraint ${\rm tr}_1[C]=\ii$ to the SDP (\ref{deltasdpinproof1})--(\ref{deltasdpinproof3}); note that the set of completely positive trace-preserving maps is compact, so that then the optimum in (\ref{eq:delta}) and (\ref{deltasdpinproof1}) is attained. {\it(b)} In the definition of the functional $\Delta(\widetilde{T})$ in (\ref{eq:distancemeasure}) one could add a term $||\widetilde{T}^*(\ii)-\ii||$ that penalizes non-trace-preserving maps and let the infimum (\ref{eq:delta}) still range over all completely positive maps. As we will see below, the infimum in (\ref{eq:delta}) is attained in case {\it(b)} as well. We also mention that, if the additional penalty term is weighted by $w>0$, then {\it(b)} goes over to {\it(a)} by formally setting $w=\infty$ (see below).

Among all of those choices, for definiteness consider the following way to formulate this question as an SDP. Choose a Hermitian basis $\{X_i\}_i$ of $\cS$, define $Y_i:=T(X_i)$, and choose a weight $w>0$.
Then define
\bea\Delta&:=&\inf\big\{\Delta(\widetilde{T})\,\big|\,\widetilde{T}:\cM_d\ra\cM_{d'}~\text{completely~positive}\big\}~,\label{definedeltaforprimal}\\
\text{where}\quad\Delta(\widetilde{T})&:=&w||\widetilde{T}^*(\ii)-\ii||\,+\,\sum_i||\widetilde{T}(X_i)-Y_i||_1~.\label{deltaTforcpandtp}\eea
Here, we chose the operator norm in the Heisenberg picture and the trace norm in the Schr\"odinger picture. As mentioned in Section \ref{sec:intro}, for given in- and outputs $X_i, Y_i$, the value $\Delta$ can be seen as a quantitative measure of the least amount of non-Markovianity necessary to implement the evolution $X_i\mapsto Y_i$ (for other approaches to this question, cf.~\cite{assessing,breuermarkov}).

This optimization can be phrased as an SDP in the following way (cf.~(\ref{deltasdpinproof1})--(\ref{deltasdpinproof3})):
\bea\Delta~=&\text{infimum~of}~&w\lambda\,+\,\sum_i\tr{P_i+Q_i}~,\label{primaloftracepreservingsdp1}\\
&\text{subject~to}~&C,\,P_i,\,Q_i\,\geq\,0~,\\
&&P_i-Q_i\,=\,{\rm tr}_2[C(\ii\otimes X_i^T)]-Y_i\quad\forall i~,\label{primaloftracepreservingsdp3}\\
&&-\lambda\ii\,\leq\,{\rm tr}_1[C]-\ii\,\leq\,\lambda\ii~.\label{primaloftracepreservingsdp4}\eea
The dual of this SDP is (cf.~(\ref{CPSDPminimizingfunction})--(\ref{onepossiblecompactnesscondition})):
\bea\Gamma\,=\,-\Delta=&\text{infimum~of}~&\tr{H_0}\,+\,\sum_i\tr{Y^T_iH_i}~,\label{dualoftracepreservingsdp1}\\
&\text{subject~to}~&H_0\otimes\ii\,+\,\sum_i X_i\otimes H_i\,\geq\,0~,\label{dualoftracepreservingsdp2}\\
&&||H_i||\,\leq\,1\quad\forall i\geq1~,\label{dualoftracepreservingsdp3}\\
&&||H_0||_1\,\leq\,w~.\label{dualoftracepreservingsdp4}\eea

Note that here, unlike in Section \ref{sectapproxextensions} generally, Slater's constraint qualification condition does not only hold for the primal problem (\ref{primaloftracepreservingsdp1})--(\ref{primaloftracepreservingsdp4}) but also for the dual (\ref{dualoftracepreservingsdp1})--(\ref{dualoftracepreservingsdp4}), e.g.~by choosing $H_i=0$ and $H_0=\epsilon\ii$ for small $\epsilon>0$. This means that, besides the strong duality $\Delta=-\Gamma$ indicated in (\ref{dualoftracepreservingsdp1}), the optimum of the primal (\ref{primaloftracepreservingsdp1}) is always attained, as claimed above and in contrast to Section \ref{sectapproxextensions} (cf.~Theorems \ref{th:cpisapprox} and \ref{unboundednesstrueapprox}). In other words, this implies that if a map $T$ has arbitrarily good approximations with completely positive trace-preserving extensions, then it actually has a completely positive trace-preserving extension $\widetilde{T}$, whereas Example \ref{ex:noexactcpextension} shows this to be generally false for the situation without trace-preservation. The existence of a completely positive trace-preserving extension can thus be checked by checking whether the value of the SDP is $\Delta=0$ or $\Delta>0$.

Whereas the latter observations constitute simplifications compared to the merely completely positive situation of Section \ref{sectapproxextensions}, an existence result for completely positive trace-preserving extensions analogous to Theorem \ref{th:exactcpextension} unfortunately does not hold here. This is demonstrated by the following example.

\begin{example}[No completely positive trace-preserving extension]\label{ex:notracepcpextension}
For $p\in[2/3,1)$, whose exact value we will specify later, define $\rho:={\rm diag}(p,1-p)$ and $\cS:={\rm span}\{\rho,\sigma_y\}\subset\cM_2$. Note that $\rho$ is strictly positive, and thus $\cS={\rm span}\{\rho,\rho-\sqrt{p(1-p)}\sigma_y/2\}$ is spanned by (strictly) positive operators.
Now define the linear map $T$,
\be T:\cS\ra\cM_2\,,\,~X\mapsto T(X):=K^{1/2}XK^{1/2}\qquad\text{with}~\,K\,:=\,\frac{1}{2}\left(\begin{matrix}3p-1&2\\2&3p+2\end{matrix}\right)~,\nonumber\ee
noting that $K^{1/2}=(K^{1/2})^\dagger=(K^{1/2})^T$ is Hermitian and symmetric as $K$ is real and positive for $p\geq2/3$. Due to its Kraus representation (viewed as a map on $\cM_2$), $T$ is completely positive and linear on $\cS$.
Moreover, $T$ is trace-preserving on $\cS$, which can be easily seen by noting $\tr{T(X)}=\tr{KX}$ and verifying $\tr{KX}=\tr{X}$ for $X=\rho,\sigma_y$.
Thus, $T$ is a completely positive trace-preserving linear map on a subspace $\cS$ which is spanned by strictly positive operators.

Yet, for some values of $p$, $T$ does not possess a positive trace-preserving extension to all of $\cM_2$ (thus, in particular no completely positive trace-preserving extension), as we show now.
Applying the identity $S\sigma_yS^T=(\det S)\sigma_y$, valid for all $S\in\cM_2$, to the symmetric matrix $S=K^{1/2}$, we get
\be T(\sigma_y)\,=\,K^{1/2}\sigma_yK^{1/2}\,=\,\sigma_y\sqrt{\det K}\,=\,\sigma_y\sqrt{(9p^2+3p-6)/4}~.\label{sigmaynoncontractive}\ee
For $p\in[2/3,1)$ large enough, e.g.~$p=14/15$, the square root in the last expression is larger than $1$, such that $||T(\sigma_y)||>||\sigma_y||$ for any norm $||\cdot||$. But a positive trace-preserving map defined on a matrix algebra has to be a contraction w.r.t.~the trace norm \cite{ruskainonincreasing}. Thus, $T$ cannot have a positive trace-preserving extension to all of $\cM_2$, as $T$ itself is not even contractive on $\cS$.

It is interesting to compare this situation to the analogous situation in the Heisenberg picture, where the property of a map being positive and unital (the Heisenberg picture analogue of trace-preservation) on an operator system also does not imply contractivity of the map (\cite{paulsen}, Chapter 2). In contrast to our corresponding situation above, however, if a map is \emph{completely} positive and unital on an operator system, then it \emph{is} contractive, since it has a (completely) positive and unital extension to the whole space by Theorem \ref{th:exactcpextension}, so that the Russo-Dye Theorem ensures contractivity (\cite{paulsen}, Corollary 2.9).

In summary, for $p=14/15$, the map $T$ does have a trace-preserving extension to $\cM_2$, and it does have a completely positive extension to $\cM_2$, but it does not have an extension to $\cM_2$ which is both completely positive and trace-preserving.
\end{example}

Easier counterexamples to the existence of completely positive trace-preserving extensions can be found if one does not require $\cS$ to contain a strictly positive operator. For example, one can take the situation from Example \ref{ex:easytocheckcp}, choose any quantum state $P\in\cM_{d'}$ and any traceless Hermitian $B\in\cM_{d'}$ with $||B||_1>||\sigma_x||_1$, e.g.~$B=c\sigma_x$ with any $c>1$ for $d'=2$. However, because in such cases not even a completely positive extension (i.e.~without the requirement of trace-preservation) might exist, as in Example \ref{ex:noexactcpextension}, this constitutes a weaker counterexample to the extendability of completely positive trace-preserving maps. If, on the other hand, $\cS$ contains a strictly positive operator as in Example \ref{ex:notracepcpextension}, then by Theorem \ref{th:exactcpextension} at least a completely positive extension is guaranteed to exist.

\medskip

Generalizing the quest for a completely positive trace-preserving extension of a map $T$ given on a subspace $\cS$, one can consider the situation outlined at the end of subsection \ref{linearityandcompatibilityprecond} where one wants to find a map $\widetilde{T}:\cM_d\ra\cM_{d'}$ which is fixed on a subspace of $\cM_d$ by $\widetilde{T}(X_i)=Y_i$ and whose dual map $\widetilde{T}^*$ is fixed on a subspace of $\cM_{d'}$ by $\widetilde{T}^*(X'_j)=Y'_j$.
If this map $\widetilde{T}$ is to be completely positive, then the problem can obviously again be phrased as an SDP like (\ref{primaloftracepreservingsdp1})--(\ref{primaloftracepreservingsdp4}), and its dual will be similar to (\ref{dualoftracepreservingsdp1})--(\ref{dualoftracepreservingsdp4}).

Concerning this situation we only note that, as shown by Example \ref{ex:notracepcpextension}, even if ${\rm span}\{X_i\}_i$ resp.~${\rm span}\{X'_j\}_j$ are spanned by strictly positive matrices and if the two maps given on those two respective subspaces are each completely positive and if the two maps are compatible in the sense that a linear extension $\widetilde{T}:\cM_d\ra\cM_{d'}$ exists, then it may still be impossible to choose this extension $\widetilde{T}$ to be completely positive.
On the other hand, the stated conditions are clearly necessary for the existence of such a completely positive extension.


\medskip

In Section \ref{albertiuhlmannsection} we will continue to examine trace-preserving completely positive extensions, for the special case of two specified in- and outputs ($N=2$).

\subsection{Extensions as probabilistic quantum operations}\label{subsectprobabextensions}

Less restrictively than in subsection \ref{cptpsubsection}, here we do not demand the quantum operation $\widetilde{T}$ to output the desired outcomes deterministically for each of the given inputs, but rather allow the operation to fail with some probability while still signalling its failure in this case. The following discussion will be restricted to the Schr\"odinger picture.

Given input and output quantum states $\rho_i\in\cM_d$ and $\rho'_i\in\cM_{d'}$ ($i=1,\ldots,N$), we ask for a quantum operation $\widetilde{T}:\cM_d\ra\cM_{d'}$ that achieves the transformation $\rho_i\mapsto\rho'_i$ probabilistically. This means that $\widetilde{T}$ should be completely positive and trace-non-increasing and satisfy $\widetilde{T}(\rho_i)=p_i\rho'_i$ for some numbers $p_i$. Each $p_i$ will be non-negative by the positivity of $\widetilde{T}$ and will not exceed $1$ and can thus be interpreted as the probability of success of the quantum operation $\widetilde{T}$ upon input of $\rho_i$; the requirement $p_i=1$ for all $i$ would correspond to the situation of the previous subsection.

This interpretation as success probabilities becomes evident since by
\be\widehat{T}(\rho):=\widetilde{T}(\rho)\otimes\ket{1}\bra{1}+\tr{(\ii-\widetilde{T}^*(\ii))\rho}\rho'_0\otimes\ket{0}\bra{0}\nonumber\ee
one can define a quantum channel that returns $1$ on the ancilla system if the quantum operation $\widetilde{T}$ on the main system was successful, which happens with probability ${\rm tr}\big[\widetilde{T}(\rho_i)\big]=p_i$ upon input of the state $\rho_i$. Note that if $\widetilde{T}$ conforms to the above conditions but is not trace-non-increasing, the latter can be achieved by rescaling $\widetilde{T}':=\widetilde{T}/||\widetilde{T}^*(\ii)||$ and rescaling the $p_i$ accordingly.

\medskip

Here we point out that the quest for a probabilistic quantum operation $\widetilde{T}$ as described above can be formulated as a semidefinite program (SDP) as well. For given $\rho_i$ and $\rho'_i$, we treat the probabilities $p_i$ as optimization variables and require as equality conditions of the SDP that
\be\widetilde{T}(\rho_i)~=~p_i\rho'_i\qquad\forall i~.\label{exactmappinginprobabsdp}\ee
Again, as for example in (\ref{primaloftracepreservingsdp3}), the action of the linear map $\widetilde{T}$ should also here be expressed in terms of its associated Choi matrix $C$, cf.~(\ref{choicorrespondenceeqn}). Complete positivity of $\widetilde{T}$ is again ensured by the SDP constraint $C\geq0$, which will also ensure non-negativity of the $p_i$, which in turn however could be phrased independently as the SDP constraint $p_i\geq0$. The trace-non-increasing condition can be phrased as the SDP constraint
\be {\rm tr}_1[C]~\leq~\ii~.\ee

Through the equality constraints (\ref{exactmappinginprobabsdp}) we have demanded that the operation $\widetilde{T}$ yields the exact desired output states, albeit not necessarily deterministically. Thus, now the objective function to optimize cannot quantify deviations from the desired output states, as before in (\ref{deltaTforcpandtp}) or (\ref{eq:distancemeasure}). Instead, we will optimize some function of the success probabilities $\{p_i\}_i$, maximizing these with respect to some figure of merit. One obvious choice is to maximize the minimum success probability $\min\{p_i\}_i$ of any transition $\rho_i\mapsto\rho'_i$. This can be achieved by maximizing an optimization variable $q$ subject to the constraint that it be smaller than any of the probabilities $p_i$:
\bea p_{min}&=&\text{supremum~of}~~~~q~,\\
&&\text{subject~to}\,~~~~~~~\,q\leq p_i\quad\forall i~,\eea
and subject to the other constraints from above. Another possible choice is to maximize the overall success probability of the operation $\widetilde{T}$, when some prior probability distribution $\{\pi_i\}_i$ of the occurrence of each of the distinguished input states $\rho_i$ is given:
\bea p_{post}~=&\text{supremum~of}~&\sum_i\pi_ip_i~.\label{weightedsuccessprobab}\eea

A trivial satisfying assignment of the SDPs above is the map $T\equiv0$, i.e.~$C=0$, which implies $p_i=0$ and thus zero success probability on any of the given input states. If this is the optimal solution of the given problem, one may say that the problem does not have a solution. More stringently, one may say that the problem does not have a solution if one of the success probabilities is necessarily zero, i.e.~if $p_{min}=0$. As may be especially desirable in the case of the objective function (\ref{weightedsuccessprobab}), one can ensure a certain minimum success probability $p_0$ for every of the input states $\rho_i$ by additionally requiring the SDP constraints $p_i\geq p_0~\forall i$, which may however make the SDP infeasible.

If one requires all success probabilities to be equal and nonzero, i.e.~$p_i\equiv p >0~~\forall i$ in (\ref{exactmappinginprobabsdp}), then the above problem becomes equivalent to the existence of a completely positive extension $T'$ of the map $\rho_i\mapsto\rho'_i$ to the full algebra $\cM_d$ as considered in Section \ref{sectapproxextensions} and in subsection \ref{sec:spannedbydensityops}. Namely, if such a completely positive extension $T'$ exists, then one can define $\widetilde{T}:=T'/||T'^*(\ii)||$ to obtain a trace-non-increasing solution to the above problem with positive output probabilities $p_i\equiv p=1/||T'^*(\ii)||>0$. Note also that, since the input space $\cS={\rm span}\{\rho_i\}_i$ is spanned by density operators as in subsection \ref{sec:spannedbydensityops}, the existence of such a map $T'$ can be efficiently checked by $\Delta=0$ in the SDP (\ref{eq:delta}) (or equivalently by $\Gamma=0$ in the SDP (\ref{CPSDPminimizingfunction})--(\ref{onepossiblecompactnesscondition})). And conversely, if $\widetilde{T}$ is a solution to the above problem with positive success probability $p_i\equiv p>0~~\forall i$, then $T':=\widetilde{T}/p$ is a completely positive extension of $\rho_i\mapsto\rho'_i$. Thus, both existence problems are equivalent up to a rescaling.

\medskip

For the case of two given in- and output density matrices ($N=2$) and again allowing for distinct success probabilities $p_1$ and $p_2$, one can find necessary and sufficient conditions for the existence of a suitable quantum operation $\widetilde{T}$ that has nonzero success probability on each of the input states $\rho_1,\rho_2$ \cite{hilbertmetricpaper}. First, it is obviously necessary that the inclusion of the supports ${\rm supp}[\rho_1]\subseteq{\rm supp}[\rho_2]$ then implies the corresponding inclusion ${\rm supp}[\rho'_1]\subseteq{\rm supp}[\rho'_2]$ for the images, and similarly for the indices $1$ and $2$ interchanged. With this \emph{compatibility condition} one gets the following result:
\begin{proposition}[Existence of a probabilistic operation; Theorem 21 in \cite{hilbertmetricpaper}]\label{propfromhilbertmetricpaper}If the aforementioned compatibility condition for the supports of quantum states $\rho_1,\rho_2\in\cM_d$ and $\rho'_1,\rho'_2\in\cM_{d'}$ is satisfied, then the existence of a quantum operation $\widetilde{T}:\cM_d\ra\cM_{d'}$ with $\widetilde{T}(\rho_i)=p_i\rho'_i$ for some strictly positive probabilities $p_i>0$ is equivalent to
\be||\rho_1^{-1/2}\rho_2^{}\rho_1^{-1/2}||\cdot||\rho_2^{-1/2}\rho_1^{}\rho_2^{-1/2}||\,~\geq\,~||{\rho'_1}^{-1/2}{\rho'_2}^{}{\rho'_1}^{-1/2}||\cdot||{\rho'_2}^{-1/2}{\rho'_1}^{}{\rho'_2}^{-1/2}||~.\label{hilbertmetricpapercondition}\ee
\end{proposition}
Condition (\ref{hilbertmetricpapercondition}), where we define $0^{-1/2}:=\infty$ and $0\cdot\infty:=0$, can be expressed more cleanly in terms of \emph{Hilbert's projective metric} \cite{hilbertmetricpaper}.

\begin{example}[Unequal vs.~equal success probabilities; connection to unambiguous state discrimination]In order to illustrate the ideas of this subsection, we show that equal success probabilities $p_i\equiv p$ pose in general a stronger requirement than allowing for distinct probabilities (even if requiring the latter to be nonzero). This is true already for the case $N=2$, in dimensions $d=d'=2$, and even for classical (i.e.~diagonal) states.

Define $\rho_1:=\rho'_2:={\rm diag}(1/3,2/3)$, $\rho_2:={\rm diag}(1/5,4/5)$, and $\rho'_1:={\rm diag}(1/2,1/2)$. Then $\widetilde{T}(\rho):=\bra{0}\rho\ket{0}{\rm diag}(1,0)+\bra{1}\rho\ket{1}{\rm diag}(0,1/2)$ defines a completely positive trace-non-increasing map, which satisfies $\widetilde{T}(\rho_i)=p_i\rho'_i$ for strictly positive probabilities $p_1=2/3$, $p_2=3/5$ (the existence of such a map $\widetilde{T}$ is guaranteed already by Proposition \ref{propfromhilbertmetricpaper}). However, for equal success probabilities $p_1=p_2\equiv p>0$ such a map cannot exist, because then (complete) positivity of $\widetilde{T}$ would imply $0\leq\widetilde{T}({\rm diag}(1,0))=\widetilde{T}(6\rho_1-5\rho_2)=p(6\rho'_1-5\rho'_2)=p\,{\rm diag}(4/3,-1/3)$, which is a contradiction. Thus, the map $\rho_i\mapsto\rho'_i$ can in particular not be achieved deterministically ($p_1=p_2=1$) by any quantum operation.

As one can also see from this example, where the $\rho_i$ have both full support, for the existence of a probabilistic map $\rho_i\mapsto\rho'_i$ it is not necessary for the input states $\{\rho_i\}_i$ to be discriminable unambiguously \cite{Che00}. Unambiguous state discrimination would correspond to choosing mutually orthogonal output states $\{\rho'_i\}_i$. And in this case the necessary compatibility condition in Proposition \ref{propfromhilbertmetricpaper} forces that neither of the supports of $\rho_1$ or $\rho_2$ be contained in the support of the other state, which is exactly the condition for unambiguous state discrimination of two states.
\end{example}

\subsection{Completely positive extensions for Abelian range or domain}\label{commutativesection}

It is well-known in operator algebra theory \cite{paulsen} that a positive linear map $T:\cA\ra\cB$ between unital $C^*$-algebras $\cA$ and $\cB$ is completely positive if the domain algebra $\cA$ or the range algebra $\cB$ is commutative. This implication from positivity to complete positivity is another flavor to our extension problem due to extension theorems like Arveson's (Theorem \ref{th:exactcpextension}; cf.~\cite{paulsen}), which require complete positivity whereas similar extensions fail for general positive maps.

In classical statistical theory the states are probability distributions, which live in commuting algebras. A commuting domain or range thus corresponds, respectively, to a classical input into or a classical output from the quantum device we are to construct (cf.~Introduction, Section \ref{sec:intro}). In the present section we will mostly restrict our attention to cases where the given inputs $\rho_i$ and outputs $\rho'_i$ are states (classical or quantum) or at least positive operators rather than general observables. This restriction will give the strongest results (cf.~for example Corollary \ref{corollaryspannedbydensitymatrices} vs.~Theorem \ref{th:cpisapprox}), and it also includes the special case of operator systems that are of special interest in operator theory \cite{paulsen}.

\bigskip

Let us first treat the case where the range algebra $\cB$ is commutative, i.e.~where the output states $\rho'_i\in\cM_{d'}$ commute pairwise and can thus be embedded into a commuting algebra $\cB$. $\cB$ could for example be the algebra of $d'\times d'$-matrices  that are diagonal in a common eigenbasis of the $\rho'_i$, which themselves can thus be interpreted as probability distributions. In this case we can formulate the following observation (cf.~Corollary \ref{corollaryspannedbydensitymatrices}):
\begin{corollary}[Positive map with commuting range]\label{commutingrange}Let $\rho_i\in\cM_d$  be quantum states and $\rho'_i\in\cM_{d'}$ be classical (i.e.~commuting) states ($i=1,\ldots,N$), and assume that $\rho_i\mapsto\rho'_i$ defines a positive linear map $T$ on $\cS:={\rm span}\{\rho_i\}_i$. Then there exists a completely positive linear extension $\widetilde{T}:\cM_d\ra\cB$ into a commuting subalgebra $\cB\subseteq\cM_{d'}$.
\end{corollary}
\begin{proof}$T:\cS\ra\cB$ is a positive map into a commuting algebra $\cB$, where we construct $\cB$ as above. For a commuting range $\cB$, the implication mentioned at the beginning of this subsection holds for maps defined on arbitrary subspaces $\cS\subseteq\cA\equiv\cM_d$ \cite{paulsen}, and thus $T:\cS\ra\cB\subseteq\cM_{d'}$ is completely positive. Corollary \ref{corollaryspannedbydensitymatrices} now gives a completely positive extension $\widehat{T}:\cM_d\ra\cM_{d'}$. We concatenate this with the (completely positive) map $\Pi:\cM_{d'}\ra\cB$ that projects each $X\in\cM_{d'}$ onto its diagonal w.r.t.~$\cB$ (and therefore leaves all $\rho'_i$ invariant) to finally get $\widetilde{T}:=\Pi\circ\widehat{T}$.
\end{proof}
Concerning Corollary \ref{commutingrange} we remark that, even under the assumption of a commuting range, a completely positive and trace-preserving map $T:\rho_i\mapsto\rho'_i$ may fail to have a positive and trace-preserving extension $\widetilde{T}:\cM_d\ra\cM_{d'}$. That is, extension theorems for quantum channels as desired in subsection \ref{cptpsubsection} do not exist even for commuting range. An example for this failure is the following: Let $\Pi:\cM_2\ra\cM_2$, $\Pi(X):=P_{+}XP_{+}+P_{-}XP_{-}$, be the quantum channel which projects onto the spectral projections $P_{+},P_{-}$ of the Pauli matrix $\sigma_y$. If $T$ and $\cS$ are taken from Example \ref{ex:notracepcpextension}, then the map $\Pi\circ T$ is completely positive and trace-preserving on $\cS$, and $\Pi\circ T(\cS)$ is commutative as each element is diagonal in the eigenbasis of $\sigma_y$. However, due to $\Pi(\sigma_y)=\sigma_y$ we have from Eq.~(\ref{sigmaynoncontractive}) still that $\Pi\circ T$ is not contractive for $p=14/15$, and thus cannot possess a quantum channel extension.

\bigskip

Now we consider maps $T:\cS\ra\cM_{d'}$ where the domain $\cS={\rm span}\{\rho_i\}_i\subseteq\cM_d$ is spanned by a commuting set of states ($i=1,\ldots,N$). Looking at a common eigenbasis, we can interpret these input states as probability vectors $\rho_i\in\R^d$. Therefore the set of normalized states in $\cS$ is the intersection ${\mathcal{P}}:={\rm aff}\{\rho_i\}_i\cap\R_+^d$ of the nonnegative orthant with the affine plane spanned by the input states. Thus, ${\mathcal{P}}$ is a convex polytope with a finite number $e$ of extreme points $\varepsilon^\alpha=\sum_ic_i^\alpha\rho_i$, $\alpha=1,\ldots,e$ ($c_i^\alpha\in\R$). By definition, the map $T$ is positive if $T(\rho)\geq0$ for all $\rho\in{\mathcal{P}}$. With knowledge of the extreme points $\varepsilon^\alpha$, however, this is equivalent to saying that all of their images are positive, i.e.
\be T(\varepsilon^\alpha)~=~\sum_ic_i^\alpha\rho'_i~\geq~0\quad\forall\alpha=1,\ldots,e~,\nonumber\ee
which is a \emph{finite} number of conditions as opposed to the infinitely many conditions in Definition \ref{definepositivemap} for a general non-commutative domain. If the extreme points $\varepsilon^\alpha$ can be found efficiently \cite{matheiss}, then this provides an efficient criterion for the positivity of $T$, which is generally not available. Below, however, we will see that positivity of $T$ on a commuting domain $\cS$ does generally not imply complete positivity.

After this prelude on checking positivity of maps with commutative domains, we can state extension results for this case.
\begin{proposition}\label{extendcommutativedomain}Let $\cS:={\rm span}\{\rho_i\}_i\subseteq\cM_d$ be spanned by a commuting set of states $\rho_i$ ($i=1,\ldots,N$), and let $T:\cS\ra\cM_{d'}$ be a positive linear map. Then:
\begin{itemize}
\item[(a)]If $d\leq3$ then there exists a completely positive extension $\widetilde{T}:\cM_d\ra\cM_{d'}$ of $T$. For $d\geq4$ there does not in general exist a positive extension $\widetilde{T}:\cM_d\ra\cM_{d'}$.
\item[(b)]If $d\leq2$ and $T$ is trace-preserving, then there exists a quantum channel $\widetilde{T}:\cM_d\ra\cM_{d'}$ extending $T$. For $d\geq3$ there does not in general exist a positive trace-preserving extension $\widetilde{T}:\cM_d\ra\cM_{d'}$.
\end{itemize}
\end{proposition}
\begin{proof}As above, we can fix an orthonormal basis $\{\ket{x}\}_{x=1}^d$ in which all $\rho_i$ are simultaneously diagonal and view these states as probability vectors in $\C^d\supseteq\cS$. The idea for the existence proofs in {\it(a)} and {\it(b)} is now to find a linear map $\Pi:\C^d\ra\cS$ which projects $\C^d$ into $\cS$, such that $\Pi$ is positive (and in case {\it(b)} furthermore trace-preserving) and leaves $\cS$ elementwise invariant, i.e.~$\Pi(\rho_i)=\rho_i$. Then the concatenated map $T\circ\Pi:\C^d\ra\cM_{d'}$ is positive and extends $T$ from $\cS$ to $\C^d$, which is a commutative unital $C^*$-algebra. As mentioned at the beginning of subsection \ref{commutativesection}, this implies that $T\circ\Pi$ is completely positive \cite{paulsen}. Now, the map $D:\cM_d\ra\C^d$, $D(X):=\sum_x\ket{x}\bra{x}X\ket{x}\bra{x}$, which projects onto the basis in which all $\rho_i$ are diagonal, is completely positive and trace-preserving and leaves each $\rho_i$ invariant. Thus, $\widetilde{T}:=T\circ\Pi\circ D:\cM_d\ra\cM_{d'}$ is a completely positive extension of $T$ (and, in case {\it(b)}, trace-preserving).

For the existence proofs, it now remains to find an appropriate map $\Pi$, which we do by exhaustion of cases. If ${\rm dim}\,\cS=d$, then $\cS=\C^d$ and $\Pi:={\rm{id}}_{\C^d}$ does the job. If ${\rm dim}\,\cS=0$, then we do not apply the above strategy but can directly choose any quantum channel $\widetilde{T}:\cM_d\ra\cM_{d'}$. If ${\rm dim}\,\cS=1$, i.e.~$\cS=\C\cdot\rho_1$, we choose $\Pi(X):=\tr{X}\rho_1$, which is positive and trace-preserving. This completes the existence proof for part {\it(b)}, whereas for part {\it(a)} only the case $d=3$, ${\rm dim}\,\cS=2$ is left, i.e.~$\cS={\rm span}\{\rho_1,\rho_2\}$ with distinct states $\rho_1\neq\rho_2\in\R_+^3$. W.l.o.g.~we may assume $\rho_1,\rho_2$ to lie on the boundary of $\R_+^3$, i.e.~to be the two extremal points of ${\mathcal{P}}$ above, so that $\rho_1={\rm diag}(p,1-p,0)$ and $\rho_2={\rm diag}(q,0,1-q)$ with $p,q\in[0,1)$. We then define $\Pi$ via $\Pi({\rm diag}(1,0,0)):={\rm diag}(0,0,0)$, $\Pi({\rm diag}(0,1,0)):=\rho_1/(1-p)$, $\Pi({\rm diag}(0,0,1)):=\rho_2/(1-q)$, so that $\Pi$ is positive and leaves $\rho_1,\rho_2$ invariant.

To show the second part of {\it(b)}, choose the states $\rho_1,\rho_2\in\cM_3$ from the previous paragraph with $p=q=1/2$ and define a linear map $T:\cS\ra\cM_2$ by $T(\rho_1):=\ket{0}\bra{0}$, $T(\rho_2):=\ket{1}\bra{1}$. $T$ is trace-preserving and positive as it maps the extremal points of $\mathcal{P}$ (see above) to quantum states. Suppose now that $\widetilde{T}:\cM_3\ra\cM_2$ is a positive and trace-preserving extension of $T$. Then $\ket{0}\bra{0}=\widetilde{T}(\rho_1)=\widetilde{T}({\rm diag}(1,0,0))/2+\widetilde{T}({\rm diag}(0,1,0))/2$ is a non-trivial convex decomposition of the pure state $\ket{0}\bra{0}$ into states, which implies $\widetilde{T}({\rm diag}(1,0,0))=\ket{0}\bra{0}$. Similarly, $\ket{1}\bra{1}=\widetilde{T}(\rho_2)=\widetilde{T}({\rm diag}(1,0,0))/2+\widetilde{T}({\rm diag}(0,0,1))/2$ implies $\widetilde{T}({\rm diag}(1,0,0))=\ket{1}\bra{1}$, which contradicts the preceding conclusion. The same contradiction obtains for any $d\geq3$ by embedding.

For the second part of {\it(a)} choose similarly $\rho_1:={\rm diag}(1,1,0,0)/2$, $\rho_2:={\rm diag}(0,1,1,0)/2$, $\rho_3:={\rm diag}(0,0,1,1)/2$ (so that $\cS\subset\cM_4$ is even an Abelian operator system), and define $T:\cS\ra\cM_2$ linearly by $T(\rho_1):=\ket{0}\bra{0}$, $T(\rho_2):=\ket{+}\bra{+}$, $T(\rho_3):=\ket{1}\bra{1}$, where $\ket{\pm}:=(\ket{0}\pm\ket{1})/\sqrt{2}$ (note, $T$ is even trace-preserving on $\cS$). The extreme points of the corresponding set ${\mathcal{P}}$ (see above) are $\rho_1$, $\rho_2$, $\rho_3$ together with $\rho_4:=\rho_1+\rho_3-\rho_2$, so that $T$ is indeed positive due to $T(\rho_4)=\ket{-}\bra{-}\geq0$. Assume now that $\widetilde{T}:\cM_4\ra\cM_2$ is a positive extension of $T$. Then, with $\rho_5:={\rm diag}(0,1,0,0)$, $\widetilde{T}(\rho_5)$ is positive and appears in non-trivial convex decompositions of the two distinct pure states $\widetilde{T}(\rho_1)$ and $\widetilde{T}(\rho_2)$, so that necessarily $\widetilde{T}(\rho_5)=0$. This however implies that the image $\widetilde{T}(\rho_6)=2\ket{1}\bra{1}-2\ket{+}\bra{+}$ of the state $\rho_6:=2\rho_3+\rho_5-2\rho_2$ is not positive, which contradicts the positivity of $\widetilde{T}$.
\end{proof}

In situations where no completely positive extension exists in part {\it(a)} of Proposition \ref{extendcommutativedomain}, the map $T$ is not completely positive on $\cS$, despite being positive on $\cS$. The reason is that, if $T$ had been completely positive on $\cS={\rm span}\{\rho_i\}_i$, it would have a completely positive extension according to Corollary \ref{corollaryspannedbydensitymatrices}. This shows that positivity on a commuting domain (even on an Abelian operator system $\cS\subset\cM_4$, cf.~proof of Proposition \ref{extendcommutativedomain}) does not imply complete positivity, contrary to the analogous statement for commuting range (Corollary \ref{commutingrange}).

\section{Alberti-Uhlmann Theorem revisited}\label{albertiuhlmannsection}

As we have seen in previous sections, we do not in general have a closed expression which decides whether a given map $T:\cS\ra\cM_{d'}$ has a completely positive and trace-preserving extension, i.e.~whether this map can be extended to a quantum channel. Whereas the semidefinite program (SDP) from subsection \ref{cptpsubsection} can answer this question and find the ``best'' approximating quantum channel, existence results in the sprit of Theorem \ref{th:exactcpextension} do not seem to be available for the quantum channel case (see Example \ref{ex:notracepcpextension}), not even under the additional assumption of commuting range or domain (see subsection \ref{commutativesection}). And even in the just completely positive case it seems generally hard to find closed criteria for the existence of an extension -- note for example that the existence result in Theorem \ref{th:exactcpextension} already presupposes complete positivity of the given map $T$ on $\cS$, which needs to be checked with an SDP in general.

The problem remains even for quantum systems of small dimensions $d,d'\geq2$ and even for small numbers $N\geq2$ of given input (and output) states. This section will exclusively treat the $N=2$ case, always requiring trace-preservation; in this section we will denote the quantum channel extension by $T$. For the case of $d=d'=2$ and $N=2$ Alberti and Uhlmann have found necessary and sufficient conditions for the existence of a quantum channel extension:
\begin{theorem}[Alberti-Uhlmann Theorem \cite{albertiuhlmann}]\label{restateautheorem}Let $\rho_1,\rho_2\in\cM_d$, $\rho'_1,\rho'_2\in\cM_{d'}$ be density matrices of at most qubit size, i.e.~$d,d'\leq2$. Then the existence of a quantum channel $T$ satisfying $T(\rho_i)=\rho'_i$ for $i=1,2$ is equivalent to (the ``Alberti-Uhlmann condition'')
\be||\rho_1-t\rho_2||_1~\geq~||\rho'_1-t\rho'_2||_1\qquad\forall t>0~.\label{aucondition}\ee
\end{theorem}
The Alberti-Uhlmann condition (\ref{aucondition}) is obviously necessary for the existence of such a quantum channel $T$ as the trace-norm is non-increasing under any positive trace-preserving map \cite{ruskainonincreasing}. Also observe that condition (\ref{aucondition}) is automatically satisfied (with equality) for all $t\leq0$, since in this case $\rho_1-t\rho_2\geq0$ such that $||\rho_1-t\rho_2||_1=\tr{\rho_1-t\rho_2}=1-t$ and equally for the output states. Alberti and Uhlmann state Theorem \ref{restateautheorem} specifically for $d=d'=2$ \cite{albertiuhlmann}. Note also that the Alberti-Uhlmann condition may equivalently be written as \cite{cjw}
\be||p_1\rho_1-p_2\rho_2||_1~\geq~||p_1\rho'_1-p_2\rho'_2||_1\qquad\text{for all probability distributions}~(p_1,p_2)~;\label{statediscrinterpret}\ee
this in turn is equivalent to the statement that, for any choice of prior probabilities $(p_1,p_2)$, the error probability in optimal state discrimination \cite{hellstroem} between the original states $\rho_1,\rho_2$ is smaller or equal than the error probability for the processed states $\rho'_1,\rho'_2$.

\medskip

The Alberti-Uhlmann condition (\ref{aucondition}) is not straightforward to check directly as it involves infinitely many inequalities, one for each $t>0$. Our first result below (subsection \ref{fidelitycritsubsection}) replaces this by a simpler criterion which involves checking only a few conditions. This new result also extends the reach of the Alberti-Uhlmann Theorem to any (finite) output dimension $d'$. Furthermore, our result can be used to give a proof \cite{bm11} of Theorem \ref{restateautheorem} which appears shorter and more transparent than the original one in \cite{albertiuhlmann}.

Remarkably, the analogue of Theorem \ref{restateautheorem} for classical probability distributions holds in \emph{any} dimensions $d,d'$ \cite{ruchetal,torgersenbook}  (i.e.~for \emph{diagonal} quantum states $\rho_i\in\cM_d$, $\rho'_i\in\cM_{d'}$), whereas Alberti and Uhlmann's proof of Theorem \ref{restateautheorem} is very specific to the qubit case $d=d'=2$ and very different from the classical proof. The obvious question of whether Theorem \ref{restateautheorem} also holds for dimensions $d,d'>2$ is met in the literature with claims \cite{albertiuhlmann,cjw} of the existence of counterexamples to this; however, to the best of our knowledge, no such counterexample can be found in the literature. In subsection \ref{counterexpsection} we will give such a counterexample, namely matrices $\rho_i,\rho'_i\in\cM_3$ such that condition (\ref{aucondition}) does not suffice for the existence of a quantum channel satisfying $T(\rho_i)=\rho'_i$; this disproof will use the SDP formulation of previous sections (specifically of subsection \ref{cptpsubsection}) in an analytic way, whereas our initial counterexample has been found numerically with the SDP.  In subsequent discussions \cite{bm11}, counterexamples stronger than ours have been found to this supposed generalization of Theorem \ref{restateautheorem} beyond the qubit case.

We further remark that, for arbitrary Hilbert space dimensions $d,d'$, the Alberti-Uhlmann Theorem can be related to the existence of certain quantum \emph{statistical morphisms}, which were introduced in \cite{buscemicmp} and which are mappings more general than positive trace-preserving maps. Namely, the formulation of a notion of statistical sufficiency for quantum experiments leads to existence questions for statistical morphisms, which is then related to conditions similar to (\ref{aucondition}) \cite{buscemicmp,matsumotorandomizationpaper,jencovamorph}.

\subsection{Fidelity criterion replacing the Alberti-Uhlmann condition (\ref{aucondition})}\label{fidelitycritsubsection}
The basic idea behind the following theorem to decide the existence of a quantum channel satisfying $T(\rho_i)=\rho'_i$ for $i=1,2$ is the following: Given qubit density matrices $\rho_1,\rho_2\in\cM_2$, there exist numbers $a,b\in[0,1]$ such that $\rho_1-a\rho_2$ and $\rho_2-b\rho_1$ are both positive and have rank at most $1$, i.e.~each is a multiple of a pure quantum state $\psi_1,\psi_2$. But for \emph{pure} input states, the existence of a quantum channel achieving the mapping $\psi_1\mapsto(\rho'_1-a\rho'_2)/(1-a)$ and $\psi_2\mapsto(\rho'_2-b\rho'_1)/(1-b)$ is easy to decide via the monotonicity of the fidelity \cite{uhlmannfidelitycrit}, and is equivalent to the existence of $T$ above (the cases $a=1$, $b=1$ will be included by separate considerations).

To formulate the theorem concisely, it is useful to introduce two notions for positive matrices $A,B\in\cM_n$, $n\geq1$, $A,B\geq0$. First, even if $A$ or $B$ are not normalized to have unit trace, we define the \emph{(generalized) fidelity} \cite{uhlmanndefinefidelity}:
\be F(A,B)~:=~\tr{\sqrt{A^{1/2}BA^{1/2}}\,}~.\nonumber\ee
Secondly, define \cite{hilbertmetricpaper} (assuming $B\neq0$)
\be\inf(A/B)~:=~\sup\{\lambda\,\big|\,A-\lambda B\geq0\}~,\nonumber\ee
and note that $A-B\inf(A/B)$ is positive with one eigenvalue $0$. $\inf(A/B)$ can be computed as $\inf(A/B)=||A^{-1/2}BA^{-1/2}||^{-1}$ if ${\rm supp}[B]\subseteq{\rm supp}[A]$ (and $\inf(A/B)=0$ otherwise) \cite{hilbertmetricpaper}. Now we can formulate the theorem.

\begin{theorem}[Fidelity criterion for the existence of a quantum channel extension]\label{ourfidelitycrition}Let $\rho_1,\rho_2\in\cM_d$ be density matrices of at most qubit size, i.e.~$d\leq2$, and let $\rho'_1,\rho'_2\in\cM_{d'}$ be density matrices ($d'\geq1$). Then the existence of a quantum channel $T$ satisfying $T(\rho_i)=\rho'_i$ for $i=1,2$ is equivalent to the three conditions
\bea&\rho'_1-a\rho'_2~\geq~0~,\quad\text{and}~~~\rho'_2-b\rho'_1~\geq~0~,\label{onetwocondfidelitycrit}\\
&\text{and}~~~F(\rho'_1-a\rho'_2,\rho'_2-b\rho'_1)~\geq~F(\rho_1-a\rho_2,\rho_2-b\rho_1)~,\label{thirdcondfidelitycrit}\eea
where we have denoted
\be a~:=~\inf(\rho_1/\rho_2)~,\quad b~:=~\inf(\rho_2/\rho_1)~.\nonumber\ee
\end{theorem}
\begin{remark}Note that condition (\ref{thirdcondfidelitycrit}) can only be evaluated if the two conditions (\ref{onetwocondfidelitycrit}) hold. Condition (\ref{thirdcondfidelitycrit}) bears some similarity to the formulation (\ref{statediscrinterpret}) of the Alberti-Uhlmann condition. Note, however, that (\ref{onetwocondfidelitycrit})--(\ref{thirdcondfidelitycrit}) constitute only a finite number of conditions compared to the infinitely many conditions (\ref{aucondition}).\end{remark}
\begin{proof}Due to $0\leq\tr{\rho_1-a\rho_2}=1-a$ we have $a\in[0,1]$, and similar for $b$. Furthermore, the rank of $\rho_1-a\rho_2\in\cM_2$ is at most $1$ since it has one zero eigenvalue as remarked above, and similar for $\rho_2-b\rho_1$. We first show sufficiency of the fidelity criterion (\ref{onetwocondfidelitycrit})--(\ref{thirdcondfidelitycrit}).

Assume for now that $\rho_1\neq\rho_2$. This implies $a,b\in[0,1)$, as $a=1$ would mean $\rho_1-\rho_2\geq0$ which together with $\tr{\rho_1-\rho_2}=0$ would mean $\rho_1-\rho_2=0$; similar for $b$. Due to $\tr{\rho_1-a\rho_2}=1-a>0$, $\rho_1-a\rho_2$ has at least rank $1$, and thus $\psi_1\equiv\ket{\psi_1}\bra{\psi_1}:=(\rho_1-a\rho_2)/(1-a)$ is a pure quantum state; similarly $\psi_2:=(\rho_2-b\rho_1)/(1-b)$. If conditions (\ref{onetwocondfidelitycrit}) and (\ref{thirdcondfidelitycrit}) hold, then $\sigma'_1:=(\rho'_1-a\rho'_2)/(1-a)$ and $\sigma'_2:=(\rho'_2-b\rho'_1)/(1-b)$ are quantum states which satisfy $F(\sigma'_1,\sigma'_2)\geq F(\psi_1,\psi_2)$. By Theorem 4.1 of \cite{uhlmannfidelitycrit}, this implies the existence of a quantum channel $T:\cM_2\ra\cM_{d'}$ satisfying $T(\psi_i)=\sigma'_i$ for $i=1,2$. The existence of this $T$ can also be derived from more well-known facts in quantum information theory, as we show in the following paragraph.

By the characterization of the fidelity as the maximal overlap of purifications \cite{uhlmanndefinefidelity,nielsenchuang}, there exist purifications of $\sigma'_i$, i.e.~pure quantum states $\ket{\psi'_i}\in\C^{d'}\otimes\C^{d'}$ with ${\rm tr}_2[\,\ket{\psi'_i}\bra{\psi'_i}\,]=\sigma'_i$, which satisfy
\be\big|\langle\psi'_1|\psi'_2\rangle\big|~=~F(\sigma'_1,\sigma'_2)~\geq~F(\psi_1,\psi_2)~=~\big|\langle\psi_1|\psi_2\rangle\big|~,\nonumber\ee
where the last equality follows from the definition of the fidelity; this is the step where we need purity of the input states $\psi_i$. Thus, there exist quantum states $\ket{\varphi'_i}\in\C^2$ (or in $\C^n$, for any given $n\geq2$) such that
\be\big|\langle\psi_1|\psi_2\rangle\big|~=~\big|\langle\psi'_1|\psi'_2\rangle\big|\cdot\big|\langle\varphi'_1|\varphi'_2\rangle\big|~=~\big|\big(\bra{\psi'_1}\otimes\bra{\varphi'_1}\big)\,\big(\ket{\psi'_2}\otimes\ket{\varphi'_2}\big)\big|~.\nonumber\ee
Due to this invariant overlap, the following defines an isometry $V:\C^2\ra\C^{d'}\otimes\C^{d'}\otimes\C^2$:
\be V\ket{\psi_i}~:=~\ket{\psi'_i}\otimes\ket{\varphi'_i}\qquad\text{for}\,~i=1,2~.\nonumber\ee
Finally, we can take this isometry to be the Stinespring dilation \cite{nielsenchuang} of the quantum channel:
\be T:\cM_2\ra\cM_{d'}\,,\quad T(X)~:=~{\rm tr}_{2,3}[VXV^\dagger]~.\nonumber\ee
Noting that $\ket{\psi'_i}\otimes\ket{\varphi'_i}\in\C^{d'}\otimes\C^{d'}\otimes\C^2$ are purifications of the $\sigma'_i$ gives $T(\psi_i)=\sigma'_i$, as desired.

This map $T$ is the desired quantum channel since by linearity
\be T(\rho_1)~=~T\left(\frac{(1-a)\psi_1+a(1-b)\psi_2}{1-ab}\right)~=~\frac{(1-a)\sigma'_1+a(1-b)\sigma'_2}{1-ab}~=~\rho'_1~,\nonumber\ee
and similarly $T(\rho_2)=\rho'_2$.

If $\rho_1=\rho_2$ then $a=1$, so that (\ref{onetwocondfidelitycrit}) implies $\rho'_1-\rho'_2\geq0$, which means $\rho'_1=\rho'_2$ due to $\tr{\rho'_1-\rho'_2}=0$. Thus the map $T(X):=\rho'_1\tr{X}$ suffices.

Conversely, if a quantum channel satisfies $T(\rho_i)=\rho'_i$, then (\ref{onetwocondfidelitycrit}) holds as $\rho'_1-a\rho'_2$ is the image of the positive matrix $\rho_1-a\rho_2$ under $T$, and similarly $\rho'_2-b\rho'_1$. And writing the channel in a Stinespring dilation $T(X)=:{\rm tr}_2[VXV^\dagger]$ with an isometry $V:\C^2\ra\C^{d'}\otimes\C^n$ \cite{nielsenchuang} shows that $V(\rho_1-a\rho_2)V^\dagger$ resp.~$V(\rho_2-b\rho_1)V^\dagger$ are purifications of $\rho'_1-a\rho'_2$ resp.~$\rho'_2-b\rho'_1$, whereas $\rho_1-a\rho_2$ and $\rho_2-b\rho_1$ are already pure (strictly speaking, as above one may want to treat the case $a=b=1$ separately and normalize the states otherwise by $1/(1-a)$ resp.~$1/(1-b)$). Thus,
\be F(\rho'_1-a\rho'_2,\rho'_2-b\rho'_1)~\geq~\tr{\big(V(\rho_1-a\rho_2)V^\dagger\big)^\dagger\,V(\rho_2-b\rho_1)V^\dagger}^{1/2}~=~F(\rho_1-a\rho_2,\rho_2-b\rho_1)\nonumber\ee
as desired, where the first inequality follows again from the fact that the fidelity is the maximum overlap over all purifications \cite{uhlmanndefinefidelity,nielsenchuang}. 
\end{proof}

Although the criteria (\ref{onetwocondfidelitycrit})--(\ref{thirdcondfidelitycrit}) make sense for input dimensions $d>2$ as well, they are not sufficient to guarantee the existence of a quantum channel $T$ achieving the mapping. The reason is that $\rho'_1-a\rho'_2$ and $\rho'_2-b\rho'_1$ may generally have rank $d-1>1$ and thus are not (multiples of) pure states, so that the fidelity criterion from \cite{uhlmannfidelitycrit} is not sufficient (see proof above).

Theorem \ref{ourfidelitycrition} can be used to give a more transparent proof \cite{bm11} of Theorem \ref{restateautheorem}. Namely, one can show \cite{bm11} that for qubits ($d=d'=2$) the Alberti-Uhlmann condition (\ref{aucondition}) implies the conditions (\ref{onetwocondfidelitycrit})--(\ref{thirdcondfidelitycrit}), which then by Theorem \ref{ourfidelitycrition} implies the existence of the desired quantum channel $T$. This two-step procedure to show sufficiency in the Alberti-Uhlmann Theorem appears conceptually simpler and shorter than the original proof in \cite{albertiuhlmann}.

\subsection{Counterexamples to higher-dimensional generalizations of Theorem \ref{restateautheorem}}\label{counterexpsection}
As mentioned above, the analogue of the Alberti-Uhlmann Theorem for classical probability distributions holds in any dimensions $d,d'\geq1$ \cite{ruchetal}: If $\mu_i\in\R^d$, $\mu'_i\in\R^{d'}$ are probability distributions satisfying the condition $||\mu_1-t\mu_2||_1\geq||\mu'_1-t\mu'_2||_1$ for all $t>0$, then there exists a stochastic matrix $M\in\R^{d'\times d}$ with $M\mu_i=\mu'_i$. This theorem is itself a generalization of a famous theorem in majorization theory by Hardy, Littlewood and P\'olya \cite{hlp}, equivalent to the case where $\mu_2=\mu'_2$ are the uniform distributions in dimension $d=d'$.

Note that, analogous to quantum channels, stochastic matrices are the most general memoryless physical maps between classical states, i.e.~between probability distributions. They are trace-preserving and completely positive (which is equivalent to positivity on the Abelian algebra of probability distributions) and can be extended to quantum channels on the larger algebra of density matrices.

With these analogies and with the classical theorem cited above \cite{ruchetal} in mind, one might be led to conjecture that the Alberti-Uhlmann Theorem (Theorem \ref{restateautheorem}) holds beyond the qubit case, i.e.~beyond the case $d,d'\leq2$. However, in the literature \cite{albertiuhlmann,cjw} there is the claim of the existence of counterexamples to this supposed generalization at least for the case $d=d'\geq3$, whereas to the best of our knowledge no such counterexamples are stated anywhere. We will remedy this here by exhibiting a counterexample in the qutrit case $d=d'=3$ and describing extensions to higher dimensions. This counterexample was obtained and is being proved with techniques coming from the SDPs of previous sections. Later, we will mention counterexamples \cite{bm11} that are stronger than the ones presented now.

\begin{proposition}[Counterexample to higher-dimensional generalizations of Theorem \ref{restateautheorem}; demonstration of numerical and analytical SDP techniques]\label{counterexampleprop}Define qutrit density matrices
\be\rho_1~:=~\frac{1}{6}\left(\begin{matrix}2&1&0\\1&2&1\\0&1&2\end{matrix}\right),\quad\rho_2~:=~\frac{1}{6}\left(\begin{matrix}2&1&0\\1&2&-i\\0&i&2\end{matrix}\right),\quad\rho'_1~:=~\rho_1^T~,~\quad\rho'_2~:=~\rho_2^T~.\label{explicitmatricescounterexmp}\ee
Then the $\rho_i,\rho'_i$ satisfy the Alberti-Uhlmann condition (\ref{aucondition}), but there does not exist a quantum channel $T$ satisfying $T(\rho_i)=\rho'_i$ for $i=1,2$.
\end{proposition}
\begin{proof}Note that the transposition map $\Theta:\cM_3\ra\cM_3,\,\Theta(X):=X^T$, is positive and trace-preserving. As the trace-norm is non-increasing under the application of positive trace-preserving maps \cite{ruskainonincreasing}, we immediately have for all $t$
\be||\rho_1-t\rho_2||_1~\geq~||\Theta(\rho_1-t\rho_2)||_1~=~||\rho'_1-t\rho'_2||_1~,\nonumber\ee
which is condition (\ref{aucondition}).

We will now prove the second claim, namely that no completely positive trace-preserving map $T:\cM_3\ra\cM_3$ satisfies $T(\rho_i)=\rho'_i$. Therefore, assume such a map exists, and we will bring this to a contradiction. For the Hermitian matrices $H_i\in\cM_3$ exhibited in Appendix \ref{appendixproof}, we have that
\be M~:=~H_0\otimes\ii_3\,+\,\rho_1\otimes H_1\,+\,\rho_2\otimes H_2~\geq~0\nonumber\ee
is positive (this can be checked numerically; the following argument also works when $H_0$ is replaced by $H_0+\epsilon\ii_3$ for $\epsilon\in(0,0.7)$, which ensures that $M$ is strictly positive with smallest eigenvalue $\geq\epsilon$). As $T$ is completely positive, $(T\otimes{\rm id}_3)(M)$ will be positive and in particular $\bra{\Omega_3}(T\otimes{\rm id}_3)(M)\ket{\Omega_3}$ should be nonnegative, where $\ket{\Omega_3}:=\sum_{i=1}^3\ket{ii}$ is the unnormalized maximally entangled state. But
\bea\bra{\Omega_3}(T\otimes{\rm id}_3)(M)\ket{\Omega_3}&=&\bra{\Omega_3}\big(T(H_0)\otimes\ii_3+T(\rho_1)\otimes H_1+T(\rho_2)\otimes H_2\big)\ket{\Omega_3}\nonumber\\
&=&\tr{T(H_0)}\,+\,\tr{\rho'_1H_1^T}\,+\,\tr{\rho'_2H_2^T}\nonumber\\
&\leq&-2.2~<~0\nonumber\eea
is negative, where we have used trace-preservation $\tr{T(H_0)}=\tr{H_0}$ of $T$ and the explicit numerical values given in (\ref{explicitmatricescounterexmp}) and in Appendix \ref{appendixproof}. This is the desired contradiction. Note that the last computation is similar to the proof of Theorem \ref{th:CPnessviaSDPtheorem} and also to the computation (\ref{deltaplusgammaobjective}), which analytically shows any dual feasible value to be a bound on the primal optimum of an SDP.

For dimensions $d,d'\geq3$, we can embed the above matrices $\rho_i,\rho'_i,H_i$ as upper-left blocks into matrices $\widetilde{\rho}_i,\widetilde{H}_0\in\cM_d$ and $\widetilde{\rho}'_i,\widetilde{H}_1,\widetilde{H}_2\in\cM_{d'}$, with all other entries vanishing. Then $||\widetilde{\rho}_1-t\widetilde{\rho}_2||_1=||\rho_1-t\rho_2||_1$ (and the primed analogue of this equation) ensure condition (\ref{aucondition}), whereas the proof by contradiction goes through as above.\end{proof}

The outputs $\rho'_i$ in this counterexample are constructed as images of the inputs $\rho_i$ under a trace-preserving positive, albeit not completely positive, map. Thus, although the mapping $\rho_i\mapsto\rho'_i$ cannot be achieved by a map from the class of quantum channels, it can be achieved by a map from the somewhat bigger class where the contraint of complete positivity is lifted to positivity. Arguably, this existence of a positive trace-preserving extension makes our counterexample somewhat weak, but is a convenient way to enforce the Alberti-Uhlmann condition (\ref{aucondition}).

To further understand why the Alberti-Uhlmann condition (\ref{aucondition}) is sufficient in $d=d'=2$ dimensions, we point out that the construction of Proposition \ref{counterexampleprop} employing positive trace-preserving maps cannot produce a counterexample in dimensions $(d,d')=(2,2),(3,2),(2,3)$. Namely, in these dimensions, any positive map $T:\cM_d\ra\cM_{d'}$ is \emph{decomposable} \cite{woronowicz}, i.e.~there exist completely positive maps $T_1,T_2:\cM_d\ra\cM_{d'}$ such that $T=T_1+T_2\Theta_d$ with the transposition $\Theta_d:\cM_d\ra\cM_d$. For $d=2$ the transposition map $\Theta_2$ acts like an inversion about the origin of the Bloch sphere followed by a unitary conjugation $U$. But the restriction of this inversion to the equatorial plane in the Bloch sphere spanned by the two input states $\rho_1,\rho_2$ has the same effect as some unitary conjugation $V$ (inducing a rotation in this plane by the angle $\pi$). Thus, for $d=2$, the mapping $\rho_i\mapsto T(\rho_i)=:\rho'_i$ ($i=1,2$) can be achieved by the completely positive and trace-preserving map $T_1+T_2UV$, i.e.~by a quantum channel. Similarly, in the case $d'=2$, we note that $T=T_1+\Theta_2T_3$ with the completely positive map $T_3:=\Theta_2T_2\Theta_d$ and that again the action of $\Theta_2$ on the two qubit states $T_3(\rho_1),T_3(\rho_2)$ can be achieved by some unitary conjugation $UV$.

The argument in the preceding paragraph suggests a two-step method to prove the Alberti-Uhlmann Theorem \ref{restateautheorem}: If one succeeds in showing that the Alberti-Uhlmann condition (\ref{aucondition}) for qubits implies the existence of a positive trace-preserving map $T$ with $T(\rho_i)=\rho'_i$, then the previous paragraph involving the Woronowicz result \cite{woronowicz} shows that this map $T$ can be chosen to be a quantum channel. The first step in this strategy, however, still seems not easy.

K.~Matsumoto \cite{bm11} has found ``strong'' counterexamples to supposed generalizations of the Alberti-Uhlmann Theorem for all dimensions $(d,d')$ not covered by Theorem \ref{restateautheorem} (i.e.~for the cases $d\geq2,d'\geq3$ and for the cases $d\geq3,d'\geq2$). These counterexamples are stronger than our Proposition \ref{counterexampleprop} in the sense that $\rho_i\in\cM_d$, $\rho'_i\in\cM_{d'}$ satisfy the Alberti-Uhlmann condition (\ref{aucondition}) whereas not even a positive trace-preserving map $T:\cM_d\ra\cM_{d'}$ achieves $T(\rho_i)=\rho'_i$ for $i=1,2$.

\bigskip

{\bf{Acknowledgments.}} The authors thank Francesco Buscemi, Keiji Matsumoto and David P\'erez-Garc\'ia for very valuable discussions. TH acknowledges support from the Academy of Finland (grant no.~138135). MAJ would like to thank for financial support the European project QUEVADIS, the Rectorate of University Politehnica Timi\c{s}oara, and the strategic grant POSDRU/21/1.5/G/13798 co-financed by the European Social Fund - Investing in People, within the Sectorial Operational Programme Human Resources Development 2007--2013. DR and MMW acknowledge support from the European projects COQUIT and QUEVADIS, the CHIST-ERA/BMBF project CQC, and from the Alfried Krupp von Bohlen und Halbach-Stiftung.

\bigskip

{\bf{Note added (22/10/2012):}} After publication of this work we became aware of the paper \cite{hlp2011}, which also concerns the existence of completely positive maps between given in- and output states and where the authors have independently obtained the condition from Theorem \ref{ourfidelitycrition} above for maps on qubits.

\appendix

\section{Strong duality in Proposition \ref{CPdualityprop}}\label{dualityproofapp}
To make the proof given in the main text self-contained, we will here use the separating hyperplane theorem from convex analysis \cite{rockafellar} to show that the optima of the two extremization problems are related by $\Delta=-\Gamma$.

For this, we abbreviate the right-hand-side of (\ref{deltasdpinproof3}) by $D_i(C):={\rm tr}_2[C(\ii\otimes X_i^T)]-Y_i$, which depends on $C$ whereas $X_i,Y_i$ are fixed. Denoting by $S_{d'}\subset\cM_{d'}$ the space of Hermitian $d'\times d'$ matrices, we will write $(z,Z_1,\ldots,Z_N)\equiv(z,Z_i)$ for elements (tuples) of the real vector space $\R\oplus(S_{d'})^{\oplus N}$, which is equipped with an inner product $\langle(z,Z_i),(w,W_i)\rangle:=zw+\sum_i\tr{Z_iW_i}$.

Note that $\Delta\in[0,\infty)$ because of (\ref{eq:delta}), so in particular $\Delta\neq\pm\infty$. Thus, we can define the following subset of $\R\oplus(S_{d'})^{\oplus N}$, where the entries of each tuple correspond to the objective function (\ref{deltasdpinproof1}) and to the equality constraints (\ref{deltasdpinproof3}) and where the quantifiers capture the inequality constraints (\ref{deltasdpinproof2}):
\be\cP\,:=\,\left\{\Big(-r+\sum_i\tr{P_i+Q_i},\,P_i-Q_i-D_i(C)\Big)\,\,\Big|\,\,C,P_i,Q_i\geq0,\,r<\Delta\right\}~.\ee
One easily verifies that $\cP$ is a convex set and that $(0,0_i)\notin\cP$, which just expresses the fact that no $C,P_i,Q_i$ which are feasible according to (\ref{deltasdpinproof2})--(\ref{deltasdpinproof3}) make the objective function (\ref{deltasdpinproof1}) smaller than $\Delta$. Therefore, the separating hyperplane theorem \cite{rockafellar} guarantees the existence of a nonzero vector $(z,Z_i)\in\R\oplus(S_{d'})^{\oplus N}$ such that $\langle(z,Z_i),(w,W_i)\rangle\geq0$ for all $(w,W_i)\in\cP$, i.e.
\be z\big(-r+\sum_i\tr{P_i+Q_i}\big)+\sum_i\tr{Z_i\big(P_i-Q_i-D_i(C)\big)}\,\geq\,0\label{fromsephyperplanethm}\ee
for all $C,P_i,Q_i\geq0$ and $r<\Delta$.

Setting $C=P_i=Q_i=0$ in (\ref{fromsephyperplanethm}) and letting $r\to-\infty$ shows $z\geq0$. Suppose now that $z=0$. Since $(z,Z_i)$ is nonzero, there exists $j$ and $\ket{\psi}$ such that $\bra{\psi}Z_j\ket{\psi}\neq0$. Set $C=P_i=Q_i=0$ for $i\neq j$, and for given $\alpha\geq0$ choose $P_j,Q_j\geq0$ such that $P_j-Q_j=-\alpha\bra{\psi}Z_j\ket{\psi}\,\ket{\psi}\bra{\psi}$. Then the left-hand-side of (\ref{fromsephyperplanethm}) becomes $-\alpha\bra{\psi}Z_j\ket{\psi}^2-\sum_i\tr{Z_iD_i(0)}$ which violates the inequality in (\ref{fromsephyperplanethm}) for $\alpha\to\infty$. Thus $z>0$.

One can therefore define $\hat{H}_i:=-Z_i^T/z$, which turn out to be optimal variables for the SDP (\ref{CPSDPminimizingfunction})--(\ref{onepossiblecompactnesscondition}) as we will see. We substitute this in (\ref{fromsephyperplanethm}), resolve $D_i(C)={\rm tr}_2[C(\ii\otimes X_i^T)]-Y_i$, rearrange terms, and use that (\ref{fromsephyperplanethm}) holds in the limit $r\nearrow\Delta$ to finally get:
\be\sum_i\tr{P_i(\ii-\hat{H}_i^T)+Q_i(\ii+\hat{H}_i^T)}\,+\,{\rm tr}\Big[C\sum_i\hat{H}_i^T\otimes X_i^T\Big]\,-\,\sum_i\tr{\hat{H}_i^TY_i}\,\geq\,\Delta\label{neweqnfromsephyperplane}\ee
for all $C,P_i,Q_i\geq0$. As the inequality (\ref{neweqnfromsephyperplane}) has to hold in particular for all $C\geq0$, one obtains $\sum_i\hat{H}_i^T\otimes X_i^T\geq0$, because otherwise $\bra{\phi}\,\sum_i\hat{H}_i^T\otimes X_i^T\,\ket{\phi}<0$ for some $\ket{\phi}$ and one could set $C:=\alpha\ket{\phi}\bra{\phi}$ to obtain a contradiction for $\alpha\to\infty$. Similar argumentation with $P_i$ and $Q_i$ gives $-\ii\leq \hat{H}_i^T\leq\ii$ for all $i$. Thus, the Hermitian matrices $\hat{H}_i$ satisfy the constraints (\ref{CPSDPconditions})--(\ref{onepossiblecompactnesscondition}). And by letting $C=P_i=Q_i=0$ in (\ref{neweqnfromsephyperplane}), one obtains an upper bound on the corresponding value of the objective function (\ref{CPSDPminimizingfunction}):
\be\sum_i\tr{Y^T_i\hat{H}_i}\,=\,\sum_i\tr{\hat{H}^T_iY_i}\,\leq\,-\Delta~,\label{objectivevalueattainedinproof}\ee
which proves $\Gamma\leq-\Delta$.

To show the other direction $\Gamma\geq-\Delta$, let $C,P_i,Q_i,H_i$ be any matrices satisfying the constraints (\ref{deltasdpinproof2})--(\ref{deltasdpinproof3}) and (\ref{CPSDPconditions})--(\ref{onepossiblecompactnesscondition}). Then the sum of the two objective functions in (\ref{deltasdpinproof1}) and (\ref{CPSDPminimizingfunction}) satisfies
\bea&&\sum_i\tr{P_i+Q_i}\,+\,\sum_i\tr{Y^T_iH_i}\label{deltaplusgammaobjective}\\
&=&\sum_i\tr{P_i(\ii-H_i^T)+Q_i(\ii+H_i^T)+H_i^TY_i}\,+\,\sum_i\tr{H_i^T(P_i-Q_i)}\nonumber\\
&=&\sum_i\tr{P_i(\ii-H_i^T)+Q_i(\ii+H_i^T)+H_i^TY_i}\,+\,\sum_i\tr{H_i^T\left({\rm tr}_2[C(\ii\otimes X_i^T)]-Y_i\right)}\nonumber\\
&=&\sum_i\tr{P_i(\ii-H_i^T)}\,+\,\sum_i\tr{Q_i(\ii+H_i^T)}\,+\,{\rm tr}\Big[C\big(\sum_iH_i\otimes X_i\big)^T\Big]\,\,\geq\,\,0\nonumber\eea
as a sum of inner products between positive matrices. Since $\Delta+\Gamma$ is the infimum of (\ref{deltaplusgammaobjective}) over all feasible variables $C,P_i,Q_i,H_i$, it is nonnegative as well, $\Delta+\Gamma\geq0$.

We have thus shown $\Delta=-\Gamma$, which also means that due to (\ref{objectivevalueattainedinproof}) the matrices $\hat{H}_i$ from above attain the optimum of the SDP (\ref{CPSDPminimizingfunction})--(\ref{onepossiblecompactnesscondition}).

\section{Witness in the proof of Proposition \ref{counterexampleprop}}\label{appendixproof}
Matrices which give the desired contradiction in the proof of Proposition \ref{counterexampleprop} are, for example,
\bea H_0~:=~\left(\begin{matrix}2.4&-5.3&0\\ *&26.7&0\\ *&*&28.8\end{matrix}\right),\quad H_1&:=&\left(\begin{matrix}10.6&-25+3.2i&+44+33.4i\\ *&54.6&-174.4-146.2i\\ *&*&44\end{matrix}\right),\nonumber\\
H_2&:=&\left(\begin{matrix}10.6&-25-3.2i&-33.4-44i\\ *&54.6&+146.2+174.4i\\ *&*&44\end{matrix}\right),\nonumber\eea
where the starred entries are to be completed so as to get Hermitian matrices. We have found these matrices numerically with MATLAB using the SeDuMi package \cite{sedumi} as (multiples of approximations to) the optimal variables in the SDP (\ref{dualoftracepreservingsdp1})--(\ref{dualoftracepreservingsdp4}), whose dual finds approximations to the given mapping $\rho_i\mapsto\rho'_i$ with a quantum channel. Note that, because in the proof we used the transposition map to define $\rho'_i:=\rho_i^T$, it was necessary to choose some non-real entries in these matrices.

Numerical experiments show that it is very easy to arrive at counterexamples as presented in Proposition \ref{counterexampleprop}: For dimensions $d=3,4$ we have sampled density matrices $\rho_i$ randomly by $\rho_i:=A_iA_i^\dagger/{\rm tr}[A_iA_i^\dagger]$, choosing the entries of each $A_i\in\cM_d$ independently and uniformly from the unit square in the complex plane, then defined $\rho'_i:=\rho_i^T$ and computed the optimal value $\Gamma=-\Delta$ of the SDP (\ref{dualoftracepreservingsdp1})--(\ref{dualoftracepreservingsdp4}). In each of our $\sim\!\!100$ trials we have found $\Delta>0$ (well within numerical precision), thereby yielding counterexamples which can be certified by a suitable assignment of the dual variables $H_i$ as in the proof of Proposition \ref{counterexampleprop}.

\bibliographystyle{alpha}

\begin{thebibliography}{WEC$^+$08}

\bibitem[Arv69]{arvesonextension}
W.~B.~Arveson, ``Subalgebras of C*-algebras'', Acta Mathematica 123, 141 (1969).

\bibitem[AU80]{albertiuhlmann}
P.~M.~Alberti, A.~Uhlmann, ``A problem relating to positive linear maps on matrix algebras'', Reports on Mathematical Physics 18, 163 (1980).

\bibitem[AU82]{connectmajorization}
P.~M.~Alberti, A.~Uhlmann, ``Stochasticity and partial order'', Mathematics and Its Applications, vol.~9, D.~Reidel Publishing Company (1982).

\bibitem[BDK$^+$05]{cleanpovms}
F.~Buscemi, G.~M.~D'Ariano, M.~Keyl, P.~Perinotti, R.~F.~Werner, ``Clean positive operator valued measures'', Journal of Mathematical Physics 46, 082109 (2005).

\bibitem[BLP09]{breuermarkov}
H.-P.~Breuer, E.-M.~Laine, J.~Piilo, ``Measure for the degree of non-Markovian behavior of quantum processes in open systems'', Physical Review Letters 103, 210401 (2009).

\bibitem[BM11]{bm11}
Correspondence with and private communication of results from Francesco Buscemi and Keiji Matsumoto (2011).

\bibitem[Bus12]{buscemicmp}
F.~Buscemi, ``Comparison of quantum statistical models:~equivalent conditions for sufficiency'', Communications in Mathematical Physics 310, 625 (2012).

\bibitem[BV04]{convexoptimization}
S.~Boyd, L.~Vandenberghe, ``Convex Optimization'', Cambridge University Press (2004).

\bibitem[Che00]{Che00}
A.~Chefles, ``Quantum state discrimination'', Contemporary Physics 41, 401 (2000).

\bibitem[Cho75]{Choi75}
M.-D.~Choi,  ``Completely Positive Linear Maps on Complex Matrices'', Linear Algebra and Its Applications 10, 285 (1975).

\bibitem[CJW04]{cjw}
A.~Chefles, R.~Jozsa, A.~Winter, ``On the existence of physical transformations between sets of quantum states'', International Journal of Quantum Information 2, 11 (2004).

\bibitem[Con90]{conwayfunctana}
J.~B.~Conway, ``A course in functional analysis'', Springer Verlag (1990).

\bibitem[Dav76]{QTOS76}
E.~B.~Davies, ``Quantum Theory of Open Systems'', Academic Press (1976).

\bibitem[Hel76]{hellstroem}
C.~W.~Hellstrom, ``Quantum detection and estimation theory'', Academic Press (1976).

\bibitem[HLP52]{hlp}
G.~H.~Hardy, J.~E.~Littlewood, G.~P\'olya, ``Inequalities'', Cambridge University Press (1952).

\bibitem[HLP$^+$12]{hlp2011}
Z.~Huang, C.-K.~Li, E.~Poon, N.-S.~Sze, ``Physical transformations between quantum states'', Journal of Mathematical Physics 53, 102209 (2012).

\bibitem[Jen11a]{jencovacp}
A.~Jencova, ``Generalized channels:~channels for convex subsets of the state space'', Journal of Mathematical Physics 53, 012201 (2012).

\bibitem[Jen11b]{jencovamorph}
A.~Jencova, ``Comparison of quantum binary experiments'', Reports on Mathematical Physics (to appear), arXiv:1110.4792 [quant-ph].

\bibitem[LP10]{commutingquantum}
C.-K.~Li, Y.-T.~Poon, ``Interpolation by completely positive maps'', Linear and Multilinear Algebra 59, 1159 (2011).

\bibitem[Mat10]{matsumotorandomizationpaper}
K.~Matsumoto, ``A quantum version of randomization criterion'', arXiv:1012.2650 [quant-ph].

\bibitem[MR80]{matheiss}
T.~H.~Matheiss, D.~S.~Rubin, ``A survey and comparison of methods for finding all vertices of convex polyhedral sets'', Mathematics of Operations Research 5, 167 (1980).

\bibitem[NC00]{nielsenchuang}
M.~A.~Nielsen, I.~L.~Chuang, ``Quantum Computation and Quantum Information'', Cambridge University Press (2000).

\bibitem[Pau02]{paulsen}
V.~Paulsen, ``Completely bounded maps and operator algebras'', Cambridge University Press (2003).

\bibitem[RKW11]{hilbertmetricpaper}
D.~Reeb, M.~J.~Kastoryano, M.~M.~Wolf, ``Hilbert's projective metric in quantum information theory'', Journal of Mathematical Physics 52, 082201 (2011).

\bibitem[Roc70]{rockafellar}
R.~T.~Rockafellar, ``Convex Analysis'', Princeton University Press (1970).

\bibitem[RSS80]{ruchetal}
E.~Ruch, S.~Schranner, T.~H.~Seligman, ``Generalization of a theorem of Hardy, Littlewood and Polya'', Journal of Mathematical Analysis and Applications 76, 222 (1980).

\bibitem[Rus94]{ruskainonincreasing}
M.~B.~Ruskai, ``Beyond strong subadditivity?~Improved bounds on the contraction of generalized relative entropy'', Reviews in Mathematical Physics, 6, 1147 (1994).

\bibitem[Stu99]{sedumi}
J.~F.~Sturm, ``Using SeDuMi 1.02, a MATLAB toolbox for optimization over symmetric cones'', Optimization Methods and Software 12, 625 (1999).

\bibitem[Tor91]{torgersenbook}
E.~Torgersen, ``Comparison of statistical experiments'', Cambridge University Press (1991).

\bibitem[Uhl76]{uhlmanndefinefidelity}
A.~Uhlmann, ``The {``transition probability''} in the state space of a *-algebra'', Reports on Mathematical Physics 9, 273 (1976).

\bibitem[Uhl85]{uhlmannfidelitycrit}
A.~Uhlmann, ``The transition probability for states of *-algebras'', Annalen der Physik 42, 524 (1985).

\bibitem[WEC$^+$08]{assessing}
M.~M.~Wolf, J.~Eisert, T.~S.~Cubitt, J.~I.~Cirac, ``Assessing non-Markovian quantum dynamics'', Physical Review Letters 101, 150402 (2008).

\bibitem[Wor76]{woronowicz}
S.~L.~Woronowicz, ``Positive maps of low dimensional matrix algebras'', Reports on Mathematical Physics 10, 165 (1976).

\end{thebibliography}

\end{document}